\documentclass[runningheads]{llncs}
\usepackage[T1]{fontenc}
\usepackage{graphicx}
\usepackage{tabularx}

\usepackage{cite}
\usepackage{centernot}
\usepackage{enumitem}
\usepackage{verbatim}
\usepackage[dvipsnames]{xcolor}
\usepackage{url}
\usepackage{xspace}
\usepackage{amsmath}
\usepackage{amsfonts}
\usepackage{listings}
\usepackage[bookmarksdepth=2]{hyperref}
\usepackage{graphicx}
\usepackage{algpseudocode}
\usepackage{algorithm}
\usepackage{mathtools}
\usepackage[mathscr]{euscript}
\usepackage{comment}
\usepackage{caption}
\usepackage{subcaption}
\usepackage{bussproofs}
\usepackage{mathpartir}
\usepackage{amssymb}
\usepackage{bbding}

\hypersetup{
 pdfborder={0 0 0},
 colorlinks=true,
 linkcolor=blue,
 urlcolor=blue,
 citecolor=blue,
}

\makeatletter
\def\orcidID#1{\href{https://orcid.org/#1}{\protect\raisebox{1.25pt}
{\protect\includegraphics{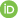}}}}
\makeatother

\algnewcommand\algorithmicforeach{\textbf{for each}}
\algdef{S}[FOR]{ForEach}[1]{\algorithmicforeach\ #1\ \algorithmicdo}

\algnewcommand\algorithmicinput{\textbf{Input:}}
\algnewcommand\algorithmicoutput{\textbf{Output:}}
\algnewcommand\Input{\item[\algorithmicinput]}
\algnewcommand\Output{\item[\algorithmicoutput]}

\algnewcommand\algorithmicproc{\textbf{Procedure}}
\algnewcommand\procedure{\item[\algorithmicproc]}

\newcommand{\update}[1]{\textcolor[RGB]{148,0,211}{#1}}

\renewcommand{\paragraph}[1]{\noindent\textbf{#1\,}}

\newcommand{\liff}{\Leftrightarrow}
\newcommand{\limpl}{\Rightarrow}
\newcommand{\strcon}[1]{\texttt{"#1"}}
\newcommand{\define}[1]{\textsl{#1}}
\newcommand{\Mo}{\mathbf{I}}
\newcommand{\Model}{\mathcal{I}}
\newcommand{\T}{\mathcal{T}}
\newcommand{\ite}{\mathit{ite}\xspace}
\newcommand{\slen}{\textsf{len}}
\newcommand{\cont}{\textsf{contains}} 
\newcommand{\true}{\mathit{True}}
\newcommand{\false}{\mathit{False}}

\newcommand{\vect}[1]{\boldsymbol{#1}}

\newcommand{\tuple}{\ensuremath\mathit{tup}}

\newcommand{\sbv}[1]{\ensuremath{\mathit{BV}_{\![#1]}}}

\newcommand{\bvult}[1]{\ensuremath<_{[#1]}}
\newcommand{\bvulte}[1]{\ensuremath\leq_{[#1]}}
\newcommand{\bvugt}[1]{\ensuremath>_{[#1]}}
\newcommand{\bvugte}[1]{\ensuremath\geq_{[#1]}}
\newcommand{\bvuprec}[1]{\ensuremath\prec_{[#1]}}
\newcommand{\bvupreccurlyeq}[1]{\ensuremath\preccurlyeq_{[#1]}}
\newcommand{\bvusucc}[1]{\ensuremath\succ_{[#1]}}
\newcommand{\bvusucccurlyeq}[1]{\ensuremath\succcurlyeq_{[#1]}}
\newcommand{\bvadd}[1]{\ensuremath+_{[#1]}}
\newcommand{\bvsub}[1]{\ensuremath-_{[#1]}}
\newcommand{\shl}[1]{\ensuremath\ll_{[#1]}}
\newcommand{\tops}{\textsc{Top}}
\newcommand{\pops}{\textsc{Pop}}
\newcommand{\solve}{\textsc{Solve}}

\newcommand{\better}{\ensuremath{\textsc{Better}}\xspace}
\newcommand{\Next}{\ensuremath{\textsc{Next}}\xspace}

\newcommand{\FV}{\ensuremath{\mathit{FV}\xspace}}
\newcommand{\mapsfrom}{\leftarrow}

\newcommand{\gomt}{GOMT\xspace}
\newcommand{\GO}{\ensuremath{{\mathcal{GO}}}\xspace}

\newcommand{\PO}{\ensuremath{{\mathcal{PO}}}\xspace}
\newcommand{\BO}{\ensuremath{{\mathcal{BO}}}\xspace}
\newcommand{\CO}{\ensuremath{{\mathcal{CO}}}\xspace}

\newif\ifarxiv
\arxivtrue 

  {\arxivfalse} 

\begin{document}
\title{Generalized Optimization Modulo Theories}
%
\titlerunning{Generalized Optimization Modulo Theories}
%
\author{}\institute{}
\author{Nestan Tsiskaridze\inst{1}\orcidID{0000-0002-4729-9770}\Envelope
\and
Clark Barrett\inst{1}\orcidID{0000-0002-9522-3084}
\and
Cesare Tinelli\inst{2}\orcidID{0000-0002-6726-775X}}

\authorrunning{N. Tsiskaridze et al.}
%
\institute{Stanford University,
  Stanford, California, USA\\
\email{\{nestan,barrett\}@cs.stanford.edu}\\
\and
The University of Iowa,
   Iowa City, Iowa, USA\\
\email{cesare-tinelli@uiowa.edu}}
%

\maketitle              
\begin{abstract}
Optimization Modulo Theories (OMT) has emerged as an important extension of the highly successful Satisfiability Modulo Theories (SMT) paradigm.  The OMT problem requires solving an SMT problem with the restriction that the solution must be optimal with respect to a given objective function.
We introduce a generalization of the OMT problem where, in particular, objective functions can range
over partially ordered sets.
We provide a formalization of and an abstract calculus for the generalized OMT problem and prove their key correctness properties. Generalized OMT extends previous work on OMT in several ways. First, in contrast to many current OMT solvers, our calculus is theory-agnostic, enabling the optimization of queries over any theories or combinations thereof. Second, our formalization unifies both single- and multi-objective optimization problems, allowing us to study them both in a single framework and facilitating the use of objective functions that are not supported by existing OMT approaches. Finally, our calculus is sufficiently general to fully capture a wide variety of current OMT approaches (each of which can be realized as a specific strategy for rule application in the calculus) and to support the exploration of new search strategies. Much like the original abstract DPLL(T) calculus for SMT, our Generalized OMT calculus is designed to establish a theoretical foundation for understanding and research and to serve as a framework for studying variations of and extensions to existing OMT methodologies.

\end{abstract}

\section{Introduction}

Over the past decade, the field of Optimization Modulo Theories (OMT) has emerged, inspiring the interest of researchers and practitioners alike. OMT builds on the highly successful Satisfiability Modulo Theories (SMT)~\cite{smt} paradigm and extends it: while the latter focuses solely on finding a theory model for a first-order formula, the former adds an objective term that must be optimized with respect to some total ordering over the term's domain.

The development of OMT solvers has fostered research across an expanding spectrum of applications, including scheduling and planning with resources~\cite{CraciunasEtAl16,KovasznaiEtAl18Puli,LeofanteEtAl21,Bit-MonnotEtAl,eracscu2020applying,feng2022reliability,jin2021joint,MarchettoEtAl21,patti2022deadline,shen2022qos},
formal verification and model checking~\cite{LiuEtAl,RatschanEtAl}, program analysis~\cite{LorenzoEtAl,DBLP:conf/vmcai/JiangCWW17,Kar7,yao2021program,HenryEtAl14}, 
requirements engineering and specification synthe\-sis~\cite{Nguyen2016RequirementsEA,NguyenEtAl17,NguyenEtAl18,GavranEtAl20}, 
security analysis~\cite{BertolissiSR18,PaolettiEtAl19,tarrach2022threat,erata2023towards}, 
system design and configuration~\cite{DemarchiEtAl19,demarchi2019automated,DBLP:conf/fmcad/TsiskaridzeSMSL21,TiernoEtAl22,rybalchenko2021supercharging,ParkEtAl20,lee2020sp,Knusel21}, 
machine learning~\cite{TESO2017166,sivaraman2020counterexample}, and quantum annealing~\cite{Bian2017SolvingSA}. 

Various OMT procedures have been developed for different types of optimization objectives (e.g., single- and multi-objective problems), underlying theories (e.g., arithmetic and bitvectors), and search strategies (e.g., linear and binary search). 
We provide an overview of established OMT techniques in Section~\ref{section:related_work}.
An extensive survey can be found in Trentin~\cite{patrickthesis}. 

We introduce a proper generalization of the OMT problem and an abstract calculus for this generalization whose main goal is similar to that of the DPLL($T$) calculus for SMT~\cite{NieOT-JACM-06}: to provide both a foundation for theoretical understanding and research and a blueprint for practical implementations. 
Our approach is general in several ways. 
First, in contrast to previous work in OMT, it is parameterized by the optimization order, which does not need to be total, and it is not specific to any theory or optimization technique, making the calculus easily applicable to new theories or objective functions.
Second, it encompasses both single- and multi-objective optimization problems, allowing us to study them in a single, unified framework and enabling combinations of objectives not covered in previous work. 
Third, it captures a wide variety of current OMT approaches, 
which can be realized as instances of the calculus together with specific strategies for rule application.  
Finally, it provides a framework for the exploration of new optimization strategies.
\medskip

\paragraph{Contributions}
To summarize, our contributions include:

\begin{itemize}
    \item a formalization of a generalization of OMT to partial orders that unifies traditional single- and multi-objective optimization problems;
    \item a theory-agnostic abstract calculus for generalized OMT that can also be used to describe and study previous OMT approaches; 
    \item a framework for understanding and exploring search strategies for generalized OMT; and
    \item proofs of correctness for important properties of the calculus.
\end{itemize}

The rest of the paper is organized as follows. Section~\ref{section:background} introduces background and notation. Section~\ref{section:framework} defines the Generalized OMT problem. Section~\ref{section:derivation_rules} presents the calculus, provides an illustrative example of its use and addresses its correctness\footnote{%
\ifarxiv
Full proofs and an additional example are provided in Appendix~\ref{appendix:proofs} and Appendix~\ref{appendix:examples}.
\else
Full proofs and an additional example are provided in a tech report~\cite{GOMT}.
\fi
}.
Finally, Section~\ref{section:related_work} discusses related work, and Section~\ref{section:conclusion} concludes.

\section{Background}
\label{section:background}

\begin{table}[t]
\footnotesize
  \centering
{%
  \begin{tabularx}{\textwidth}{l@{\hspace{.2cm}}l@{\hspace{.2cm}}X}
    \hline
    \textbf{Syntax} & \textbf{Semantics} & \textbf{Meaning} \\
    \hline
    \textsf{Bool,Int,Real,\sbv{n},Str} & &Sorts for Booleans, integers, reals, bitvectors of length $n$, and character strings \\
    $+, -, \times, \div$ & &Arithmetic operators over reals/integers \\
    $<_\textsf{R}, >_\textsf{R}, \leq_\textsf{R}, \geq_\textsf{R}$ & $\prec_\textsf{R}, \succ_\textsf{R}, \preccurlyeq_\textsf{R}, \succcurlyeq_\textsf{R}$ & Comparison operators over reals\\
    $<_\textsf{Int}, >_\textsf{Int}, \le_\textsf{Int}, \ge_\textsf{Int}$ & $\prec_\textsf{Int},  \succ_\textsf{Int}, \preccurlyeq_\textsf{Int}, \succcurlyeq_\textsf{Int}$ & Comparison operators over integers\\
    $+_{[n]}$,$-_{[n]}$,$\times_{[n]}$,$\div_{[n]}$
    & &Arithmetic modulo $2^n$ operators\\
    $\bvult{n}, \bvugt{n}$, $\bvulte{n}, \bvugte{n}$ 
    & $\bvuprec{n}, \bvusucc{n}$, $\bvupreccurlyeq{n}, \bvusucccurlyeq{n}$ &(Unsigned) comparison operators over \sbv{n} terms\\
    $x \shl{n} y$ & &Shift left operator over \sbv{n} terms\\
    $\ite(c, x, y)$ &
    & If-then-else operator (if $c$ then $x$ else $y$) \\
    $\tuple(t_1,\dots,t_n)$ &
    & $n$-ary tuple where element $i$ is $t_i$\\
    $<_{\textsf{str}}, >_{\textsf{str}}, \leq_{\textsf{str}}, \geq_{\textsf{str}}, $ & $\prec_{\textsf{str}}, \succ_{\textsf{str}}, \preccurlyeq_{\textsf{str}}, \succcurlyeq_{\textsf{str}}$ &Strict and non-strict lexicographic and reverse lexicographic orders on strings\\
    $\epsilon$ & &The empty string\\
    $x \cdot y$ & & String concatenation operator\\
    $\slen(x)$ & & String length operator\\
    $\cont(x, y)$ & & String containment operator (true iff $y$ is a substring of $x$)\\
    \hline\\
  \end{tabularx}
}
  \caption{Theory-specific notation.}
  \label{tab:theoryops}
\end{table}

We assume the standard many-sorted first-order logic setting for SMT, with the usual
notions of signature, term, formula, and interpretation.
We write $\Model\models\phi$ to mean that formula $\phi$ holds in or is \emph{satisfied} by an interpretation $\Model$. 
A \define{theory} is a pair $\T = (\Sigma, \Mo)$, where
$\Sigma$ is a signature and  $\Mo$ is a class of $\Sigma$-interpretations.
We call the elements of $\Mo$ \define{$\T$-interpretations}.
We write $\Gamma\models_\T\phi$, where $\Gamma$ is a formula (or a set of formulas), 
to mean that $\Gamma$ \emph{$\T$-entails} $\phi$, 
i.e., every $\T$-interpretation that satisfies (each formula in) $\Gamma$ 
satisfies $\phi$ as well.
For convenience, for the rest of the paper, 
\emph{we fix a background theory $\T$ with equality and with signature $\Sigma$.}
We also fix an infinite set $\mathcal{X}$ of sorted variables with sorts from $\Sigma$ and assume $\prec_\mathcal{X}$ is some total order on $\mathcal{X}$.
We assume that all terms and formulas are $\Sigma$-terms and $\Sigma$-formulas with free variables from $\mathcal{X}$.
Since the theory $\T$ is fixed, we will often abbreviate $\models_\T$ as $\models$ and consider only interpretations 
that are $\T$-interpretations assigning a value to every variable in $\mathcal{X}$.  
At various places in the paper, we use sorts and operators from standard SMT-LIB
theories such as integers, bitvectors, strings,\footnote{For simplicity,
we assume strings are over characters ranging only from \texttt{'a'} to
\texttt{'z'}.} or data types~\cite{smt-lib}.
We assume that every $\T$-interpretation interprets them in the same (standard) way.  Table~\ref{tab:theoryops} lists theory symbols used in this paper and their meanings.
A $\Sigma$-formula $\phi$ is
\define{satisfiable} (resp., \define{unsatisfiable}) \define{in $\T$}
if it is satisfied by some (resp., no) $\T$-interpretation.


Let $s$ be a $\Sigma$-term. We denote by $s^\Model$ the value of $s$ in an interpretation $\Model$, defined as usual by recursively determining the values of sub-terms. 
We denote by $\FV(s)$ the set of all variables occurring in $s$. 
Similarly, we write $\FV(\phi)$ to denote the set of all the free variables occurring in a formula $\phi$.
If $\FV(\phi) = \{v_1,\dots,v_n\}$, where for each $i\in[1,n), v_i \prec_{\mathcal{X}} v_{i+1}$, then the relation \define{defined by} $\phi$ (in $\T$) is
$\{(v_1^{\Model},\dots,v_n^{\Model}) \mid 
 \Model\models\phi \text{ for some $\T$-interpretation $\Model$} \}$.
A relation is \define{definable in $\T$} if there is some formula 
that defines it.
Let $\vect{v}$ be a tuple of variables $(v_1,\dots,v_n)$, and let $\vect{t}=(t_1,\dots,t_n)$ be a tuple of $\Sigma$-terms, such that $t_i$ and $v_i$ are of the same sort for $i\in [1,n]$; then, we denote by $s[{\vect{v} \mapsfrom \vect{t}]}$ 
the term obtained from $s$ by simultaneously replacing each occurrence of variable $v_i$ 
in $s$ with the term $t_i$.

If $S$ is a finite \emph{sequence} $(s_1,\dots,s_n)$,  we write \tops($S$) to denote, $s_1$, the first element of $S$ in $S$;  we write \pops($S$) to denote the subsequence $(s_2,\dots,s_n)$ of $S$.  We use $\emptyset$ to denote both the empty set and the empty sequence.  We write $s\in S$ to mean that $s$ occurs in the sequence $S$, and write $S \circ S'$ for the sequence obtained by appending $S'$ at the end of $S$.
%

We adopt the standard notion of \define{strict partial order} $\prec$ 
on a set $A$, that is,
a relation in $A \times A$ that is irreflexive, asymmetric, and transitive.
The relation $\prec$ is a \define{strict total order} if, in addition, 
$a_1 \prec a_2$ or $a_2 \prec a_1$ for every pair $a_1, a_2$ of distinct elements of $A$.
As usual, we will call $\prec$ \define{well-founded} over a subset $A'$ of $A$
if $A'$ contains no infinite descending chains.
An element $m \in A$ is \define{minimal (with respect to $\prec$)}
if there is no $a \in A$ such that $a \prec m$.  If $A$ has a unique minimal element, it is called a \define{minimum}.

\section{Generalized Optimization Modulo Theories}
\label{section:framework}

\noindent
We introduce a formalization of the \emph{Generalized Optimization Modulo Theories} problem which unifies single- and multi-objective optimization problems and lays the groundwork for the calculus presented in Section~\ref{section:derivation_rules}. 
 
\subsection{Formalization} 

For the rest of the paper, we fix a theory $\T$ with some signature $\Sigma$.

\begin{definition}[Generalized Optimization Modulo Theories (\gomt)]
\label{def:omt}
A \define{Generalized Optimization Modulo Theories problem} is a tuple $\GO:= \langle t, \prec, \phi \rangle$, where:
\begin{itemize}
    \item $t$, a $\Sigma$-term of some sort $\sigma$, is an \define{objective term} to optimize;

    \item $\prec$ is a \define{strict partial order definable in $\T$}, whose defining formula has two free variables, each of sort $\sigma$; and

    \item $\phi$ is a $\Sigma$-formula.

\end{itemize}

\end{definition}

\noindent
For any \gomt problem \GO and $\T$-interpretations $\Model$ and $\Model'$, we say that:
\begin{itemize}
\item $\Model$ is \GO-\define{consistent} if $\Model \models \phi$; 
\item $\Model$ \GO-\define{dominates} $\Model'$, denoted by $\Model <_{\GO} \Model'$, if $\Model$ and $\Model'$ are \GO-consistent
    and $ t^{\Model} \! \prec  t^{\Model'}$; and
\item $\Model$ is a \GO-\define{solution} if $\Model$ is \GO-consistent and no $\T$-interpretation \GO-domi\-nates $\Model$.
\end{itemize}

Informally, the term $t$
represents the \emph{objective function}, whose value we want to optimize.
The order $\prec$ is used to compare values of $t$,
with a value $a$ being considered \emph{better} than a value $a'$ if $a \prec a'$. Finally, the formula $\phi$ imposes constraints on the values that $t$ can take.
It is easy to see that the value of $t^\Model$ assigned by a
\GO-solution $\Model$ is always minimal.  As a special case, if
$\prec$ is a total order, then $t^\Model$ is also unique (i.e., it is a minimum).
Once we have fixed a \gomt problem \GO, we will informally refer 
to a \GO-consistent interpretation as a \define{solution (of $\phi$)}
and to a \GO-solution as an \define{optimal solution}.



Our notion of Generalized OMT is closely related to one by Bigarella et al.~\cite{DBLP:conf/frocos/BigarellaCGIJRS21}, which defines a notion of OMT for a generic background theory using a predicate that corresponds to a total order in that theory.
Definition~\ref{def:omt} generalizes this in two ways.  First, we allow partial orders, with total orders being a special case. 
One useful application of this generalization is the ability to model multi-objective problems as single-objective problems over a suitable partial order, as we explain below. 
Second, we do not restrict $\prec$ to correspond to a predicate symbol in the theory.  Instead, any partial order \emph{definable} in the theory can be used. This general framework captures a large class of optimization problems.

\begin{example}\label{minLRA}
Suppose $\T$ is the theory of real arithmetic with the usual signature.
Let $\GO:= \langle x+y, \prec, 0 < x\,\wedge\,xy=1
\rangle$, where $x$ and $y$ are variables of sort \textsf{Real} and
$\prec$ is defined by the formula $v_1 <_\textsf{R} v_2$ (where $v_1
\prec_\mathcal{X} v_2$). A \GO-solution is any interpretation that interprets $x$ and $y$ as 1.
\end{example}

\begin{example}\label{maxLIA}
With $\T$ now being the theory of integer arithmetic,
let $\GO= \langle x, \prec, x^2 < 20 \rangle$, where $x$ is of sort
\textsf{Int}, and $\prec$ is defined by $v_1 >_\textsf{Int} v_2$ (where $v_1
\prec_\mathcal{X} v_2$).  A \GO-solution must interpret $x$ as the maximum integer satisfying $x^2 < 20$ (i.e., $x$ must have value 4).
\end{example}

The examples above are both instances of what previous work refers to
  as single-objective optimization problems~\cite{patrickthesis}, with
  the first example being a minimization and the second a maximization
  problem.  The next example illustrates a less conventional ordering.

Note that from
now on, to keep the exposition simple, we define partial orders $\prec$
appearing in \GO problems only \emph{semantically}, i.e., formally, but
without giving a specific defining formula.  However, it is easy to check that
all orders used in this paper are, in fact, definable in a suitable $\T$.

\begin{example}\label{maxAbsLIA}
Let $\GO = \langle x, \prec, x^2 < 20 \rangle$ be a variation of Example~\ref{maxLIA}, where now, for any integers $a$ and $b$, $a \prec b$ iff $|b| \prec_{\textsf{Int}} |a|$.  A \GO-solution can interpret $x$ either as 4 or $-4$.  Neither solution dominates the other since their absolute values are equal.
\end{example}

\noindent
We next show how multi-objective problems are also instances of Definition~\ref{def:omt}.

\subsection{Multi-Objective Optimization}
\label{section:multi}

We use the term \emph{multi-objective optimization} to refer to an optimization problem consisting of several sub-problems, each of which is also an optimization problem.  A multi-objective optimization may also require specific interrelations among its sub-problems.  In this section, we define several varieties of multi-objective optimization problems and show how each can be realized using Definition~\ref{def:omt}.  For each, we also state a correctness proposition which follows straightforwardly from the definitions.

In the following, given a strict ordering $\prec$, we will denote its reflexive closure by $\preccurlyeq$. 
We start with a multi-objective optimization problem which requires that the sub-problems be prioritized in lexicographical order~\cite{nuZ14, nuZ15, OptiMathSAT, OMTLAcosts, patrickthesis}.

\begin{definition}[Lexicographic Optimization (LO)]\label{def:LO}
A \define{lexicographic optimization problem} is a sequence of \gomt problems $\mathcal{LO}=(\GO_1,\dots,\GO_n)$, where $\GO_i := \langle t_i, \prec_i, \phi_i\rangle$ for $i \in [1,n]$. 
For $\T$-interpretations $\Model$ and $\Model'$, we say that:
\begin{itemize}[leftmargin=0.7em]
\item $\Model$ $\mathcal{LO}$-\define{dominates} $\Model'$, denoted by $\Model <_{\mathcal{LO}} \Model'$,  if  
$\Model$ and $\Model'$ are $\GO_i$-\emph{consistent} for each $i\in[1,n]$,
and for some $j \! \in \! [1,n]$:
\begin{itemize}[align=parleft]
\item[(i)] $t_i^{\Model} = t_i^{\Model'}$ for all $i\in[1,j)$; and
\item[(ii)] $t_j^{\Model}$ $\prec_j$ $t_j^{\Model'}$.
\end{itemize}
\item $\Model$ is a \define{solution} to $\mathcal{LO}$ iff $\Model$ is $\GO_i$-consistent for each $i$ and no $\T$-interpretation $\mathcal{LO}$-dominates $\Model$.
\end{itemize}
\end{definition}

\noindent
An $\mathcal{LO}$ problem can be solved by converting it into an instance of Definition~\ref{def:omt}.

\begin{definition}[$\GO_\mathcal{LO}$]\label{def:OMT-LO}
    Given an $\mathcal{LO}$ problem $(\GO_1,\ldots,\GO_n)$, 
    with $\GO_i := \langle t_i, \prec_i, \phi_i\rangle$ for $i \in [1,n]$,
    the corresponding \GO instance is defined as \\
    $\GO_\mathcal{LO}(\GO_1,\ldots,\GO_n) := \langle t, \prec_{\mathcal{LO}}, \phi \rangle$, where:
    \begin{itemize}
        \item $t=\tuple(t_1,\ldots,t_n)$;
        \qquad
        $\phi = \phi_1\wedge\cdots \wedge\phi_n$;
        \item if $t$ is of sort $\sigma\!$, then 
        $\prec_{\mathcal{LO}}$ is the lexicographic extension of $(\!\prec_1, \!\ldots\!,\!\prec_n)$ to $\sigma^\T\!$: 
        
        for $(a_1,\dots,a_n)$, $(b_1,\dots,b_n)\in \sigma^\T$,  $(a_1,\dots,a_n) \prec_{\mathcal{LO}} (b_1,\dots,b_n)$ iff 
        for some $j\in[1,n]:$
        \begin{itemize}[align=parleft]
        \item[(i)] $ a_i $  $=$  $ b_i$ for all
        $i \in [1,j)\text{;  and }$
        \item[(ii)] $a_j $  $\prec_j$  $b_j$.        
        \end{itemize}
   
  
  
    \end{itemize} 
\end{definition}

\noindent
Here and in other definitions below, we use the data type theory constructor $\tuple$ to construct the objective term $t$.  This is a convenient mechanism for keeping an ordered list of the sub-objectives and keeps the overall theoretical framework simple.  In practice, if using a solver that does not support tuples or the theory of data types, other implementation mechanisms could be used.
Note that if each sub-problem uses a total order, then $\prec_\mathcal{LO}$ will also be total.
\begin{proposition}
    Let $\Model$ be a $\GO_\mathcal{LO}$-solution.  Then $\Model$ is also a solution to the corresponding $\mathcal{LO}$ problem as defined in Definition~\ref{def:LO}.
\end{proposition}

\begin{example}[$\mathcal{LO}$]\label{ex:LOBV} 
Let $\GO_1\!:=\!\langle x, \prec_1, \!\true \rangle$ and $\GO_2\!:=\! \langle y \bvadd{2} z, \prec_2, \!\true\rangle$, where  $x, y, z$ are variables of sort \sbv{2}, $a \!\prec_1 \!b$ iff $a \!\bvuprec{2}\! b$, and $a \!\prec_2 \!b$ iff $a \!\bvusucc{2}\! b$. 
Now, let $\GO = \GO_\mathcal{LO}(\GO_1,\GO_2) = \langle t, \prec_\mathcal{LO}, \true\rangle$. Then, $t=\tuple(x,y\bvadd{2}z)$ and
$(a_1,a_2) \prec_\mathcal{LO} (b_1,b_2)$ iff $a_1 \bvuprec{2} b_1$ or $(a_1 = b_1 \text{ and } a_2 \bvusucccurlyeq{2} b_2 )$.

\noindent
Now, let $\Model$, $\Model'$, and $\Model''$ be such that: $x^\Model=11, y^\Model= 00, z^\Model=10$, and $t^{\Model} :=(11,10)$; $x^{\Model'}=01, y^{\Model'}=01, z^{\Model'}=01$, and $t^{{\Model'}}:=(01,10)$; $x^{\Model''}=01, y^{\Model''}=01, z^{\Model''}= 10$, and $t^{{\Model''}}:=(01,11)$.
Then, $\Model'' <_\GO \Model'<_\GO \Model$, since $(01,11) \prec_\mathcal{LO} (01,10) \prec_\mathcal{LO} (11,10)$. 

\end{example}

\noindent
We can also accommodate Pareto optimization~\cite{nuZ14, nuZ15, patrickthesis} in our framework.

\begin{definition}[Pareto Optimization (PO)]\label{def:PO}
A \define{Pareto optimization problem} is a sequence of \gomt problems $\PO = (\GO_1,\dots,\GO_n)$, where $\GO_i := \langle t_i, \prec_i,\\ \phi_i \rangle$ for $i \in [1,n]$. For any $\T$-interpretations $\Model$ and $\Model'$, we say that:
\begin{itemize}[leftmargin=0.7em]
    \item ${\Model}$ \PO-\define{dominates}, or \define{Pareto dominates}, $\Model'$, denoted by $\Model <_{\PO} \Model'$, if $\Model$ and $\Model'$ are \GO-\emph{consistent} w.r.t. each $\GO_i$, $i\in[1,n]$,
and: 
\begin{itemize}[leftmargin=3em, align=parleft]
\item[(i)] $t_i^{\Model} \preccurlyeq_i t_i^{\Model'}$ for all $i \in [1,n]$; and
\item[(ii)] for some $j \! \in \! [1,n]$,
        $t_j^{\Model} \prec_j t_j^{\Model'}$. 
\end{itemize}

\item $\Model$ is a \define{solution} to \PO iff $\Model$ is \GO-consistent w.r.t. each $\GO_i$ and no $\Model'$ \PO-dominates $\Model$. 
\end{itemize}    
\end{definition}
\begin{definition}[$\GO_\PO$]\label{def:OMT-PO} 
Given a PO problem $\PO = (\GO_1,\dots,\GO_n)$, we define $\GO_\PO(\GO_1,\dots,\GO_n) := \langle t, \prec_{\PO}, \phi \rangle$, where: 
    \begin{itemize}
        \item $t=\tuple(t_1,\ldots,t_n)$;\qquad
        $\phi = \phi_1\wedge\cdots \wedge\phi_n$; 
        \item 
        if $t$ is of sort $\sigma$, then $\prec_{\PO}$ is the pointwise extension of 
        $(\prec_1, \ldots,\prec_n)$ to $\sigma^{\T}$;
        for any $(a_1,\dots,a_n), (b_1,\dots,b_n) \in \sigma^\T$, $(a_1,\dots,a_n) \prec_{\PO} (b_1,\dots,b_n)$ iff:
        \begin{itemize}[align=parleft]
        \item[(i)] $a_i $  $\preccurlyeq_i$  $b_i$  $ \text{for all } i \in [1,n]\text{;  and}$
        \item[(ii)]  $a_j$  $\prec_j$  $b_j$  $ \text{for some } j\in[1,n]$. 
        \end{itemize}
        
        
        
        

    \end{itemize}
\end{definition}  
\begin{proposition}
    Let $\Model$ be a $\GO_\PO$-solution.  Then $\Model$ is also a solution to the corresponding \PO problem as defined in Definition~\ref{def:PO}.
\end{proposition}

\noindent
Next, consider a \PO example with two sub-problems: one minimizing the length of a string $w$, and the other maximizing a substring $x$ of $w$ lexicographically.

\begin{example}[\PO]\label{POPLIASTR}
Let $\T$ be the SMT-LIB theory of strings and let
$\GO_1:= \langle \slen(w), \prec_1, \slen(w) < 4 \rangle$ and 
$\GO_2:= \langle x, \prec_2, \cont(w, x)\rangle$, where  $w$, $x$ are
  variables of sort \textsf{Str}, $\prec_1$ is $\prec_\textsf{Int}$, and $\prec_2$ is
$\succ_{\textsf{Str}}$.
Now, let $\GO_\PO =
\GO_\PO(\GO_1, \GO_2)\\=\langle t, \prec_\PO, \slen(w)<4 \wedge \cont(x, w) \rangle$.
Then, $t = \tuple(\slen(w), x)$ and
$(a_1,a_2) \prec_\PO (b_1,b_2)$ iff $a_1 \preccurlyeq_{\textsf{Int}} b_1, \; a_2 \succcurlyeq_{\textsf{str}} b_2$, and $(a_1  \prec_{\textsf{Int}} b_1$ or $a_2 \succ_{\textsf{str}} b_2)$. Now, let $\Model$, $\Model'$, and $\Model''$ be such that: $\Model:=\{w \mapsto \strcon{aba}, x \mapsto \strcon{ab}\}$ and $t^{\Model}:=(3,\strcon{ab})$; $\Model':=\{w \mapsto \strcon{z}, x \mapsto \strcon{z}\}$ and $t^{\Model'}:=(1,\strcon{z})$; and $\Model'':=\{w \mapsto \epsilon, x \mapsto \epsilon\}$ and $t^{\Model}:=(0,\epsilon)$.
Then, $\Model' <_{\GO} \Model$, since $(1, \strcon{z}) \prec_\PO (3, \strcon{ab})$; but both $\Model$ and $\Model'$ are incomparable with $\Model''$.  Both $\Model'$ and $\Model''$ are optimal solutions.
\end{example}

\noindent
\ifarxiv
We can similarly capture the MinMax and MaxMin optimization problems~\cite{OptiMathSAT, patrickthesis} as corresponding instances of Definition~\ref{def:omt} as described below.
\begin{definition}[MinMax]\label{def:MINMAX}
A MinMax optimization problem is a sequence of \gomt problems, $\mathcal{MINMAX} = (\GO_1,\dots,\GO_n)$, where $\GO_i := \langle t_i, \prec, \phi_i\rangle$ for $i \in [1,n]$. For any $\T$-interpretations $\Model$ and $\Model'$, we say that:
\begin{itemize}[leftmargin=0.7em]
    \item $\Model$ $\mathcal{MINMAX}$-\define{dominates} $\Model'$, denoted by $\Model <_{\mathcal{MINMAX}} \Model'$,  if $\Model$ and $\Model'$ are $\GO_i$-consistent for each $i\in[1,n]$, and 
 $t_{max}^{\Model} \prec t_{max}^{\Model'}$ where: 
 \begin{itemize}[leftmargin=3em, align=parleft]
     \item [(i)] $t_{max}^{\Model} = t_i^{\Model}$ for some $i \in [1,n]$, $t_{max}^{\Model'} = t_j^{\Model'}$ for some $j \in [1,n]$; and 
     \item[(ii)] $t_k^{\Model} \preccurlyeq t_{max}^{\Model}$ and $t_k^{\Model'} \preccurlyeq t_{max}^{\Model'}$ for all 
 $k \in [1,n]$
\end{itemize}

    \item $\Model$ is a \define{solution} to $\mathcal{MINMAX}$ iff $\Model$ is $\GO_i$-consistent for each $i$ and no $\Model'$ $\mathcal{MINMAX}$-dominates $\Model$. 
\end{itemize}
\end{definition}
\begin{definition}[$\GO_\mathcal{MINMAX}$]\label{def:OMT-MINMAX}
Given a $\mathcal{MINMAX}$ problem $(\GO_1,\dots,\GO_n)$, we define
 $\GO_\mathcal{MINMAX}(\GO_1,\dots,\GO_n) := \langle t, \prec_{\mathcal{MNMX}}, \phi \rangle$, where:
    \begin{itemize}
        \item $t = \tuple (t_1,\dots,t_n)$;\qquad
        $\phi = \phi_1 \wedge \dots \wedge \phi_n$; 
        \item if $t$ is of sort $\sigma$, then for any $(a_1,\dots,a_n),(b_1,\dots,b_n)\!\in\! \sigma^\T$, \\
        $(a_1,\dots,a_n) \!\prec_{\mathcal{MNMX}}\! (b_1,\dots,b_n)$ iff: 

        $\text{for some }i,j\!\in\![1,n]\!:$
        \begin{itemize}[align=parleft]
        \item[(i)] $a_k \! \preccurlyeq \!a_i$ $\text{\! and \!}$ $b_k \! \preccurlyeq \! b_j \text{ for all } k \!\in\! [1,n]\text{; and \!}$
        \item[(ii)] $a_i \!\prec\! b_j$.
        \end{itemize}
    \end{itemize}
\end{definition}    
\begin{proposition}
    Let $\Model$ be a $\GO_\mathcal{MINMAX}$-solution.  Then $\Model$ is also a solution to the corresponding $\mathcal{MINMAX}$ problem as defined in Definition~\ref{def:MINMAX}.
\end{proposition}

\begin{example}[$\mathcal{MINMAX}$]\label{MINMAXLRA} Let $\GO_1:=\langle a+b+c, \prec_\textsf{Int}, a\! <\! b + c\rangle$ and $\GO_2:= \langle abc, \prec_\textsf{Int}, 0 \leq abc \rangle$, where  $a$,\! $b$,\! $c$ are variables of sort \textsf{Int}.
Now, let $\GO_3=\GO_\mathcal{MINMAX}(\GO_1,\GO_2) = \langle t, \prec_\mathcal{MNMX}, a < b + c \wedge 0 \leq abc  \rangle$, where $t \!:= \!(a+b+c, abc)$ and $(a_1, a_2)\! \prec_\mathcal{MNMX} \!(b_1, b_2)$ iff $\max(a_1,a_2) \prec_{\textsf{Int}} \max(b_1,b_2)$. Let $\Model$, $\Model'$, and $\Model''$ be such that:
$\Model:=\{a \mapsto 2, b \mapsto 1, c \mapsto 3\}$ and $t^{\Model}\!\!:=(6, 6)$;
$\Model'\!\!\!:=\{a \mapsto 1, b \mapsto 0, c \mapsto 4\}$ and $t^{\Model'}\!\!\!:=(5, 0)$; and
$\Model''\!\!:=\!\!\{a \mapsto 0, b \mapsto 1, c \mapsto 4\}$ and $t^{\Model''}:=(5, 0)$.
Then, $\Model'$ and $\Model''$ both $\mathcal{MINMAX}$-dominate  $\Model$, since $(5,0) \prec_\mathcal{MNMX} (6,6)$.  But neither $\Model'$ nor $\Model''$ $\mathcal{MINMAX}$-dominates the other.
\end{example}
\noindent
The $\mathcal{MAXMIN}$ optimization problem is the dual of $\mathcal{MINMAX}$, and it can be defined in a similar way. 

\begin{definition}[MaxMin]\label{def:MAXMIN}
A MaxMin optimization problem is a sequence of \gomt problems, $\mathcal{MAXMIN} = (\GO_1,\dots,\GO_n)$, where $\GO_i := \langle t_i, \prec, \phi_i\rangle$ for $i \in [1,n]$. For any $\T$-interpretations $\Model$ and $\Model'$, we say that:
\begin{itemize}[leftmargin=1em]
    \item $\Model$ $\mathcal{MAXMIN}$-\emph{dominates} $\Model'$, denoted by $\Model <_{\mathcal{MAXMIN}} \Model'$,  if $\Model$ and $\Model'$ are $\GO_i$-\emph{consistent} for each $i\in[1,n]$, and 
    $t_{min}^{\Model'} \prec t_{min}^{\Model}$ where: 
 \begin{itemize}[leftmargin=3em, align=parleft]
    \item[(i)] $t_{min}^{\Model} = t_i^{\Model}$ for some $i \in [1,n]$, $t_{min}^{\Model'} = t_j^{\Model'}$ for some $j \in [1,n]$; and 
    \item[(ii)] $t_{min}^{\Model} \preccurlyeq t_k^{\Model}$ and $t_{min}^{\Model'} \preccurlyeq t_k^{\Model'}$ for all $k \in [1,n]$; 
 \end{itemize}

\item $\Model$ is a \define{solution} to $\mathcal{MAXMIN}$ iff $\Model$ is $\GO_i$-consistent for each $i$ and no $\Model'$ $\mathcal{MAXMIN}$-dominates $\Model$. 
\end{itemize}
\end{definition}

\begin{definition} [$\GO_\mathcal{MAXMIN}$]\label{def:OMT-MAXMIN} Given a $\mathcal{MAXMIN}$ problem $(\GO_1,\dots,\GO_n)$, we define 
$\GO_\mathcal{MAXMIN}(\GO_1,\dots,\GO_n) := \\ \langle t, \prec_{\mathcal{MXMN}}, \phi \rangle$, where:
    \begin{itemize}
        \item $t=\tuple(t_1,\dots,t_n)$; \quad $\phi = \phi_1,\dots,\phi_n$;
        \item if $t$ is of sort $\sigma$, then for any $(a_1,\dots,a_n),(b_1,\dots,b_n)\!\in\! \sigma^\T$, \\
        $(a_1,\dots,a_n) \!\prec_{\mathcal{MXMN}}\! (b_1,\dots,b_n)$ iff: 
        
        $\text{for some }i,j\!\in\![1,n]\!:$
        \begin{itemize}[align=parleft]
            \item [(i)] $a_i \!\preccurlyeq \!a_k \text{ and } b_j \!\preccurlyeq \!b_k  \text{ for all } k \!\in\! [1,n]$;  and
            \item[(ii)] $b_j \!\prec\! a_i$.  
        \end{itemize}
    \end{itemize}
\end{definition} 

\begin{proposition}
    Let $\Model$ be a $\GO_\mathcal{MAXMIN}$-solution.  Then $\Model$ is also a solution to the corresponding $\mathcal{MAXMIN}$ problem as defined in Definition~\ref{def:MAXMIN}.
\end{proposition}

\begin{example}[$\mathcal{MAXMIN}$]\label{MAXMINLRA} Let $\GO_1\!:=\!\langle a+b+c, \prec_\textsf{Int}, a\! <\! b + c\rangle$ and $\GO_2:= \langle abc, \prec_\textsf{Int}, 0 \leq abc \rangle$, where  $a$,\! $b$,\! $c$ are variables of sort \textsf{Int}.
Now, let $\GO_3=\GO_\mathcal{MAXMIN}(\GO_1,\GO_2) = \langle t, \prec_\mathcal{MXMN}, a < b + c \wedge 0 \leq abc  \rangle$, where $t \!:= \!(a+b+c, abc)$ and $(a_1, a_2)\! \prec_\mathcal{MXMN} \!(b_1, b_2)$ iff $\min(b_1,b_2) \prec_{\textsf{Int}} \min(a_1,a_2)$. Let $\Model$, $\Model'$, and $\Model''$ be such that:
$\Model:=\{a \mapsto 2, b \mapsto 1, c \mapsto 3\}$ and $t^{\Model}\!\!:=(6, 6)$;
$\Model'\!\!\!:=\{a \mapsto 1, b \mapsto 0, c \mapsto 4\}$ and $t^{\Model'}\!\!\!:=(5, 0)$; and
$\Model''\!\!:=\!\!\{a \mapsto 0, b \mapsto 1, c \mapsto 4\}$ and $t^{\Model''}:=(5, 0)$.
Then, $\Model$ $\mathcal{MAXMIN}$-dominates both $\Model'$ and $\Model''$, since $(6,6) \prec_\mathcal{MXMN} (5,0)$.  But neither $\Model'$ nor $\Model''$ $\mathcal{MAXMIN}$-dominates the other.
\end{example}

\else
Though we omit them for space reasons, we can similarly capture the MinMax and MaxMin optimization problems~\cite{OptiMathSAT, patrickthesis} as corresponding $\GO_\mathcal{MINMAX}$ and $\GO_\mathcal{MAXMIN}$ instances of Definition~\ref{def:omt}.\footnote{
Details of these formulations can be found in the longer version of this paper~\cite{GOMT}.
}
\fi

%

\ifarxiv
Note that except for degenerate cases, the orders $\prec_\mathcal{MNMX}$, $\prec_\mathcal{MXMN}$, as well as the order $\prec_\PO$ above, are always partial orders.  
\else
Note that except for degenerate cases, the orders used for MinMax and MaxMin, as well as the order $\prec_\PO$ above, are always  partial orders.  
\fi
Being able to
model these multi-objective optimization problems in a clean and simple way is
a main motivation for using a partial instead of a total order in Definition~\ref{def:omt}.

Another problem in the literature is the \emph{multiple-independent} (or \emph{boxed}) optimization problem~\cite{nuZ14, nuZ15, patrickthesis}.  It simultaneously solves several independent \gomt problems.  We show how to realize this as a single \GO instance.

\begin{definition}[Boxed Optimization (BO)] \label{def:BOX}
    A boxed optimization problem is a sequence of \gomt problems, $\BO = (\GO_1,\dots,\GO_n)$, where $\GO_i := \langle t_i, \prec_i,\\ \phi_i\rangle$ for $i \in [1,n]$. We say that:
\begin{itemize}[leftmargin=0.7em]
    \item A sequence of interpretations ${(\Model_1,\dots,\Model_n)}$ \BO-\define{dominates} $(\Model'_1,\dots,\Model'_n)$, denoted by $(\Model_1,\dots,\Model_n) <_{\BO} (\Model'_1,\dots,\Model'_n)$, if $\Model_i$ and $\Model'_i$ are $\GO_i$-\emph{consistent} or each $i\in[1,n]$,
and:
\begin{itemize}[align=parleft]
\item[(i)] $t_i^{\Model_i} \preccurlyeq_i t_i^{\Model'_i}$ for all $i \in [1,n]$; and
\item[(ii)] for some $j \! \in \! [1,n]$,
        $t_j^{\Model_j} \prec_j t_j^{\Model'_j}$.
\end{itemize}
\item $(\Model_1,\dots,\Model_n)$ is a \define{solution} to \BO iff $\Model_i$ is $\GO_i$-consistent for each $i\in[1,n]$ and no $(\Model'_1,\dots,\Model'_n)$ \BO-dominates $(\Model_1,\dots,\Model_n)$.
\end{itemize}        
\end{definition}

\noindent
Note that in previous work, there is an additional assumption that $\phi_i =
\phi_j$ for all $i,j\in[1,n]$.  Below, we show how to solve the more general case without this assumption.
We first observe that the above definition closely resembles
Definition~\ref{def:PO} for Pareto optimization (PO) problems.
Leveraging this similarity, we show
how to transform an instance of a BO problem into a PO problem.
\begin{definition}\label{def:OMT-BOX}
    ($\GO_\BO$) 
Let $\BO = (\GO_1,\dots,\GO_n)$, where
$\GO_i := \langle t_i, \prec_i, \phi_i\rangle$ for $i \in [1,n]$.
Let $V_i$ be the set of all free variables in the
$i^{th}$ sub-problem that also appear in at least one other sub-problem:
\[V_i = (\FV(t_i)\cup\FV(\phi_i)) \;\cap \bigcup_{j\in[1,n],j \neq i}
\FV(t_j)\cup\FV(\phi_j).\]
Let
$\vect{v_i}=(v_{i,1},\dots,v_{i,m})$ be some ordering of the variables in $V_i$
(say, by $\prec_\mathcal{X}$), and for each $j\in[1,m]$, let $v'_{i,j}$ be a
fresh variable of the same sort as $v_{i,j}$, and let
$\vect{v'_i}=(v'_{i,1},\dots,v'_{i,m})$.  Then, let $t'_i =
t_i[\vect{v_i} \mapsfrom \vect{v'_i}]$, $\phi'_i =
\phi_i[\vect{v_i} \mapsfrom \vect{v'_i}]$, and $\GO'_i =
\langle t'_i, \prec_i, \phi'_i\rangle$.  Then we define
$\GO_\BO:=\GO_\PO(\GO'_1,
\dots,\GO'_n)$.
\end{definition}
%
\begin{proposition}
    Let $\Model$ be a solution to $\GO_\BO$ as defined in Definition~\ref{def:OMT-BOX}.  Then $(\Model_1,\dots,\Model_n)$ is a solution to the corresponding \BO problem as defined in Definition~\ref{def:BOX}, where for each $i\in[1,n]$, $\Model_i$ is the same as $\Model$ except that each variable $v_{i,j}\in V_i$ is interpreted as $(v'_{i,j})^\Model$.
\end{proposition}

\noindent
In practice, solvers for BO
problems can be implemented without variable renaming (see, e.g.,
\cite{Li2014SymbolicOW,nuZ14,OMTLAcosts}).  Variable renaming, while a useful
theoretical construct, also adds generality to our definition of $\BO$.  An interesting direction for future experimental work would
be to compare the two approaches in practice.
\medskip

\noindent
\textbf{Compositional Optimization} \; 
\gomt problems can also be combined by functional composition of multiple objective terms, possibly of different sorts, yielding \emph{compositional optimization problems}~\cite{TESO2017166, CimattiEtAl2013,  patrickthesis}.
  Our framework handles them naturally by simply constructing an objective term capturing the desired compositional
relationship.
For example, compositional objectives can address the (partial) \text{MaxSMT} problem~\cite{patrickthesis}, where some formulas are \emph{hard} constraints and others are \emph{soft} constraints.  The goal is to satisfy all hard constraints and as many soft constraints as possible.  The next example is inspired by Cimatti et al.~\cite{CimattiEtAl2013} and Teso et al.~\cite{TESO2017166}.

\begin{example}[\textsc{MaxSMT\!}]\label{COPBLRA}
Let $x  \!\geq \! 0$ and $y \! \geq \! 0$ be hard constraints and
$4x+y-4  \!\geq\!  0$ and $2x+3y-6  \!\geq \! 0$ soft constraints. We can formalize this as $\GO_\CO=\langle t, \prec, \phi
\rangle$, where: $t=\ite(4x+y-4 \! \geq \! 0,0,1) + \ite(2x+3y-6 \! \geq \!  0,0,1)$,   $\prec\ \equiv\ \prec_\textsf{Int}$, and $\phi = x  \!\geq\!  0 \wedge y \! \geq \! 0$. An optimal solution must satisfy both hard constraints and, by minimizing the objective term $t$, as many soft constraints as possible.
\end{example}
%
\textsc{MaxSMT} has various variants including generalized, partial, weighted, and partial weighted \textsc{MaxSMT}~\cite{patrickthesis}, all of which our framework can handle similarly.

Next, we show a different compositional example that combines two different orders, one on strings and the other on integers.  This example also illustrates a theory combination not present in the OMT literature.

\begin{example}[Composition of $\textsf{Str}$ and $\textsf{Int}$]\label{COSTRLIA}
Let $\T$ be again the theory of strings\footnote{The SMT-LIB theory of strings includes the theory of integers to support constraints over string length.} 
Let $\GO_\CO = \langle \tuple(x,\slen(x))$, $\prec$, $\cont(x, \text{\strcon{a}}) \wedge \slen(x)>1\rangle$, where $x$ is of sort $\textsf{Str}$ and 
$(a_1,b_1) \prec (a_2,b_2)$ iff $b_1 \prec_{\textsf{Int}} b_2$ or $(b_1 = b_2$ and $a_1 \succ_{\textsf{str}} a_2)$. $\prec$ prioritizes minimizing the length, but then maximizes the string with respect to lexicographic order. An optimal solution must interpret $x$ as the string \strcon{za} of length 2 since $x$ must be of length at least 2 and contain \strcon{a}, making \strcon{za} the largest string of minimum length.
\end{example}

Based on the definitions given in this section, we see that our
formalism can capture any combination of 
 \GO (including compositional), $\GO_\mathcal{LO}$, $\GO_\PO$,
$\GO_\mathcal{MINMAX}$, $\GO_\mathcal{MAXMIN}$, and
$\GO_\BO$ problems.  And note that the last four all make use of the partial order feature of Definition~\ref{def:omt}.


\section{The \gomt Calculus}
\label{section:derivation_rules}

We introduce a calculus for solving the \gomt problem, presented 
as a set of \emph{derivation rules}. 
We fix a \gomt problem $\GO = \langle t, \prec, \phi \rangle$
where $\phi$ is satisfiable (optimizing does not make sense otherwise).
We start with a few definitions.
\begin{definition}[State]
A \emph{state} is a tuple $\Psi = \langle \Model, \Delta, \tau \rangle$, where $\Model$ is an interpretation,  $\Delta$ is a formula,\! and $\tau$ is a sequence of formulas.
\end{definition}

\noindent
The set of all states forms the state space for the \gomt problem. 
Intuitively, the proof procedure of the calculus is a search procedure over this state space
which maintains at all times a \emph{current state}  
$\langle \Model, \Delta, \tau \rangle$ storing a candidate solution and
additional search information.
In the current state,
$\Model$ is the best solution found so far in the search;  
$\Delta$ is a formula describing the remaining, yet unexplored, part of the state space, 
where a better solution might exist; and
$\tau$ contains formulas that divide up the search space described by $\Delta$
into \emph{branches} represented by the individual formulas in $\tau$,
maintaining the invariant that 
the disjunction of all the formulas $\tau_1, \ldots, \tau_p$ 
in $\tau$ is equivalent to $\Delta$ modulo $\phi$, 
that is, $\phi \models (\bigvee_{i=1}^p \tau_i \liff \Delta)$. 

Note that states contain $\T$-interpretations,
which are possibly infinite mathematical structures.  This is useful to keep the calculus simple. 
In practice,
it is enough just to keep track of the interpretations of the (finitely-many)
symbols without fixed meanings (variables and uninterpreted functions and sorts) appearing in the state, much as SMT solvers do in order to produce models. 

\begin{definition}[Solve] $\textsc{Solve}$ is a function that takes a formula and returns a satisfying interpretation if the formula is satisfiable and a distinguished value $\bot$ otherwise.
\end{definition}

\begin{definition}[Better]\label{def:better}
$\better_\GO$ is a function that takes a \GO-consistent interpretation $\Model$ and returns a formula $\better_\GO(\Model)$ with the property that for every \GO-consistent interpretation $\Model'$,
\[ \Model' \models \better_\GO(\Model) \text{\ iff\ \,} \Model' <_{\GO} \Model.\]
\end{definition}
\noindent
The function above is specific to the given optimization problem \GO
or, put differently, is parametrized by $t$, $\prec$, and $\phi$.
When \GO is clear, however, we simply write $\better$, 
for conciseness.

The calculus relies on the existence and computability 
of $\textsc{Solve}$ and $\better$.  $\textsc{Solve}$ can be realized by any standard SMT solver.  $\better$ relies on a defining formula for $\prec$ as discussed below. We note that intuitively, $\better(\Model)$ is simply a (possibly unsatisfiable) formula characterizing the solutions of $\phi$ that are better than $\Model$.  Assuming $\alpha_{\prec}$ is the formula defining $\prec$, with free variables $v_1 \prec_\mathcal{X} v_2$, if the value $t^\Model$ can be represented by some constant $c$ (e.g., if $t^\Model$ is a rational number), then $\better(\Model) = \alpha_{\prec}[(v_1,v_2)\mapsfrom(t,c)]$ satisfies Definition~\ref{def:better}.  On the other hand, it could be that $t^\Model$ is not representable as a constant (e.g., it could be an algebraic real number); then, a more sophisticated formula (involving, say, a polynomial and an interval specifying a particular root) may be required.

\begin{definition}[Initial State]
The \emph{initial  state} of the \gomt problem 
$\GO = \langle t, \prec, \phi \rangle$ is 
$\langle \Model_0, \Delta_0, \tau_0\rangle$, where $\Model_0 = \textsc{Solve}(\phi)$, $\Delta_0 = \better(\Model_0)$, $\tau_0 = (\Delta_0)$.
\end{definition}
\noindent

\noindent
Note that $\Model_0 \not= \bot$ since we assume that $\phi$ is satisfiable.
The search for an optimal solution to the \gomt problem in our calculus starts
with an arbitrary solution of the constraint $\phi$ and continues until
it finds an optimal one.

\begin{figure}[t]
    \centering
\begin{gather*}
\inferrule*[Left=F-Split]{
   \tau \not= \emptyset 
   \quad \psi = \textsc{Top}(\tau)
 \quad \phi \models \psi \liff \bigvee_{j=1}^k \psi_j,\; k \geq 1 
}{
   \tau := (\psi_1,\dots,\psi_k)\circ\textsc{Pop}(\tau)
}
\\[1ex]
\inferrule*[Left=F-Sat]{
   \tau \neq \emptyset
   \quad \psi = \textsc{Top}(\tau)
   \quad \textsc{Solve}(\phi\wedge\psi) = \Model' \quad \Model' \not=\bot \quad \Delta' = \Delta \wedge \better(\Model') 
}{
   \Model := \Model', \; 
   \Delta:= \Delta', \; 
   \tau := (\Delta')
}
\\[1ex]
\inferrule*[Left=F-Close]{
   \tau \neq \emptyset \quad \psi = \textsc{Top}(\tau) \quad \textsc{Solve}(\phi \wedge \psi) = \bot
}{
   \Delta := \Delta \wedge \neg \psi, \; 
   \tau := \textsc{Pop}(\tau)
}
\end{gather*}
\caption{The derivation rules of the \gomt Calculus.}
\label{rules}
\end{figure}

\subsection{Derivation Rules} 
\label{basic_derivation_rules}

Figure~\ref{rules} presents the derivation rules of the \gomt calculus.  The rules are given in guarded assignment form, where the rule premises describe the conditions on the current state that must hold for the rule to apply, and the conclusion describes the resulting modifications to the state.
State components not mentioned in the conclusion of a rule are unchanged.

A derivation rule \define{applies} to a state if 
$(i)$ the conditions in the premise are satisfied by the state and
$(ii)$ the resulting state is different.
A state is \define{saturated} if no rules apply to it.
A \GO-\emph{derivation} is a sequence of states, possibly infinite, where the first state is the initial state of the \gomt problem \GO, and each state in the sequence is obtained by applying one of the rules to the previous state.  The \define{solution sequence} of a derivation is the sequence made up of the solutions (i.e., the interpretations) in each state of the derivation. 

The calculus starts with a solution for $\phi$
and improves on it until an optimal solution is found.
During a derivation, the best solution found so far is maintained
in the $\Model$ component of the current state.
A search for a better solution can be organized into branches through the use of the \textsc{F-Split} rule.
Progress toward a better solution is enforced by the formula $\Delta$ 
which, by construction, is falsified by all the solutions found so far.
We elaborate on the individual rules next.
\medskip

\noindent
\textbf{$\textsc{F-Split}$} \,
$\textsc{F-Split}$ divides the branch of the search space represented 
by the top formula $\psi = \textsc{Top}(\tau)$ in $\tau$ 
into $k$ sub-branches $(\psi_1, \dots,\psi_k)$, ensuring their disjunction is equivalent to $\psi$ modulo the constraint $\phi$: 
$\phi \models \psi \liff \bigvee_{j=1}^k \psi_j$.
The rest of the state remains unchanged. $\textsc{F-Split}$ 
is applicable whenever $\tau$ is non-empty.
The rule does not specify how the formulas $\psi_1, \ldots, \psi_k$ are chosen. 
However, a pragmatic implementation should aim to generate them so that 
they are \emph{irredundant}
in the sense that no formula is entailed 
modulo $\phi$ by the (disjunction of the) other formulas.
This way, 
each branch potentially contains a solution that the others do not.
Note, however, that this is not a requirement.

\medskip
\noindent
\textbf{$\textsc{F-Sat}$} \,The $\textsc{F-Sat}$ rule applies when there is a solution in the branch represented by the top formula $\psi$ in $\tau$. 
The rule selects a solution $\Model' = \textsc{Solve}(\phi \wedge \psi)$ from that branch. One can prove that, by the way the formulas in $\tau$ are generated in the calculus, $\Model'$ necessarily improves on the current solution $\Model$,
moving the search closer to an optimal solution.\footnote{
\ifarxiv
See Lemma~\ref{lemma_F-Sat_better} in Appendix~\ref{appendix:proofs}.
\else
See Lemma 6 in Appendix B of a longer version of this paper~\cite{GOMT}.
\fi} 
Thus, $\textsc{F-Sat}$ switches to the new solution (with $\Model:=\Model'$) and 
directs the search 
to seek an even better solution 
by updating $\Delta$ to $\Delta' = \Delta\ \wedge\ \textsc{Better}(\Model')$.
Note that $\textsc{F-Sat}$
resets $\tau$ to the singleton sequence $(\Delta')$, 
discarding any formulas in $\tau$.
This is justified, as any discarded better solutions must also be in the space defined by $\Delta'$.
\medskip

\noindent
\textbf{$\textsc{F-Close}$} \,
The $\textsc{F-Close}$ rule eliminates the first element $\psi$ of a non-empty $\tau$
if the corresponding branch contains
no solutions (i.e., $\textsc{Solve}(\phi \wedge \psi) = \bot$).
The rule further updates the state by adding the negation of $\psi$ to $\Delta$ 
as a way to eliminate from further consideration the interpretations satisfying $\psi$.
\medskip

\noindent
Note that rules $\textsc{F-Sat}$ and $\textsc{F-Close}$ both update $\Delta$ to reflect the remaining search space, whereas $\textsc{F-Split}$ refines the division of the current search space.

\subsection{Search strategies}
\label{search_strategies}

The \gomt calculus provides the flexibility 
to support different search strategies.
Here, we give some examples, including both notable strategies from the OMT literature as well as new strategies enabled by the calculus, and explain how they work at a conceptual level.
\medskip

\noindent
\textbf{Divergence of strategies:}
The strategies discussed below, with the exception of Hybrid search, may diverge if an optimal solution does not exist
or if there is a \emph{Zeno-style}~\cite{PB-OMT12,Sebastiani:2015:OMT:2737801.2699915} infinite chain 
of increasingly better solutions, all dominated by an optimal one.  We discuss these issues and termination in general in Section~\ref{section:proofs:basic}.

\vspace{.2cm}
\noindent
\textbf{Linear search:}
A linear search strategy is obtained by never using the $\textsc{F-Split}$ rule. Instead, the $\textsc{F-Sat}$ rule is applied to completion
(that is, repeatedly until it no longer applies).
As we show later (see Theorem~\ref{theorem_basic_termination}), in the absence of Zeno chains, $\tau$ eventually becomes empty, terminating the search.
At that point, $\Model$ is guaranteed to be an optimal solution.
\medskip

\noindent
\textbf{Binary search:} 
A binary search strategy is achieved by using 
the $\textsc{F-Split}$ rule to split the search space represented 
by $\psi = \textsc{Top}(\tau)$ into two subspaces, 
represented by two formulas $\psi_1$ and $\psi_2$, 
with $\phi \models \psi \liff (\psi_1 \vee \psi_2)$.
In a strict binary search strategy, $\psi_1$ and $\psi_2$ should be chosen so that the two subspaces are disjoint and, to the extent possible, of equal size.
A typical binary strategy alternates applications of \textsc{F-Split}
with applications of either \textsc{F-Sat} or \textsc{F-Close}
until $\tau$ becomes empty,
at which point $\Model$ is guaranteed to be an optimal solution.
A smart strategy would aim to find an optimal solution as soon as possible by arranging for solutions in $\psi_1$ (which will be checked first) to be better than solutions in $\psi_2$, if this is easy to determine.
Note that an unfortunate choice of $\psi_1$ by \textsc{F-Split},
containing no solutions at all, is quickly remedied by an application
of \textsc{F-Close} which removes $\psi_1$, allowing $\psi_2$ to be considered next.  The same problem of Zeno-style infinite chains can occur in this strategy.
\medskip

\noindent
\textbf{Multi-directional exploration:} For multi-objective optimization problems, a search strategy can be defined to simultaneously direct the search space towards any or all objectives. Formally, if $n$ is the number of objectives, then the $\textsc{F-Split}$ rule can be instantiated in such a way that $\psi_j = \bigwedge_{i=1}^n \psi_{ji}$, where  $\psi_{ji}$ is a formula describing a part of the search space for the $i^{th}$ objective term in the $j^{th}$ branch.

\vspace{.2cm}
\noindent \textbf{Search order:} We formalize $\tau$ as a sequence to enforce exploring the branches in $\tau$ in a specific order, assuming such an order can be determined at the time of applying $\textsc{F-Split}$. Often, this is the case. For example, in binary search, it is typically best to explore the section of the search space with better objective values first. If a solution is found in this section, a larger portion of the search space is pruned. Conversely, if the branches are explored in another order, even finding a solution necessitates continued exploration of the space corresponding to the remaining branches.

Alternatively, $\tau$ can be implemented as a set, by redefining the \textsc{Top} and \textsc{Pop} functions accordingly to select and remove a desired element in $\tau$. With $\tau$ defined as a set, additional search strategies are possibile, including parallel exploration of the search space and the ability to arbitrarily switch between branches.

\medskip

\noindent
\textbf{Hybrid search:} 
For some objectives and orders, there exist off-the-shelf external optimization procedures (e.g., Simplex for linear real arithmetic). One way to integrate such a procedure into our calculus is 
to replace a call to the $\textsc{Solve}$ function in $\textsc{F-Sat}$
with a call to an external optimization procedure 
\textsc{Optimize} that is sort- and order-compatible with the \gomt problem.
We pass to $\textsc{Optimize}$ as parameters the constraint
$\phi \wedge \textsc{Top}(\tau)$ and the objective $t$
and obtain an optimal solution in the current branch 
$\textsc{Top}(\tau)$.\footnote{This assumes there exists an optimal solution in the current branch.  If not (i.e., if the problem is unbounded), a suitable error can be raised and the search terminated.}
The call can be viewed as an accelerator for a linear search 
on the current branch.
This approach incorporates theory-specific optimization solvers in much the same way as is done in the OMT literature. However, our calculus extends previous approaches with the ability to blend theory-specific optimization with theory-agnostic optimization 
by interleaving applications of \textsc{F-Sat} using \textsc{Solve} 
with applications using \textsc{Optimize}.
For example, we may want to alternate
between expensive calls to an external optimization solver and calls to a standard solver that are guided by a custom branching heuristic.
\medskip

\noindent
\textbf{Other strategies:} 
The calculus enables us to mix and match the above strategies arbitrarily, as well as to model other notable search techniques like cutting planes~\cite{SRI:simplex:dpllt} by integrating a cut formula into $\better$. And, of course, one advantage of an abstract calculus is that its generality provides a framework for the exploration of new strategies. Such an exploration is a promising direction for future work.

\subsection{New Applications}
\label{new_applications}

A key feature of our framework is that it is theory-agnostic, that is, it can be used with any SMT theory or combination of theories.  This is in contrast to most of the \gomt literature in which a specific theory is targeted.  It also fully supports arbitrary composition of \gomt problems using the multi-objective approaches described in Section~\ref{section:multi}.  Thus, our framework enables \gomt to be extended to new application areas requiring either combinations of theories or multi-objective formulations that are unsupported by previous approaches.
We illustrate this (and the calculus itself) using
a Pareto optimization problem over the theories of strings and integers (a combination of theories and objectives unsupported by any existing OMT approach or solver).
\begin{example}[$\GO_\mathcal{PO}$]\label{derivation_example_pareto}
Let $\GO_1:= \langle \slen(w), \prec_1, \slen(s) < \slen(w) \rangle$ and 
$\GO_2:= \langle x, \prec_2, x = s \cdot w \cdot s\rangle$, where $w$, $x$, $s$ are of sort $\textsf{Str}$, $\slen(w)$ and $\slen(s)$ are of sort $\textsf{Int}$, $\prec_1 \ \equiv \ \prec_{\textsf{Int}}$, and $\prec_2 \ \equiv \ \succ_{\textsf{Str}}$;. 
Then, let $\GO_\mathcal{PO}(\GO_1,\GO_2):=\langle t,
 \prec_\mathcal{PO}, \phi \rangle$, where $t$ is $\tuple(\slen(w), x)$, $\phi$ is $x = s \cdot w \cdot s \ \wedge\  \slen(s) < \slen(w)$, and $(a_1,a_2)  \prec_\mathcal{PO}  (b_1,b_2)$ iff $a_1  \preccurlyeq_1  b_1, \; a_2  \preccurlyeq_2  b_2$, and either $a_1 
  \prec_1  b_1$ or $a_2  \prec_2  b_2$ or both.
Suppose initially:
\[
\begin{array}{rcl}

\Model_0 & =\ & \{x \mapsto \strcon{aabaa}, s \mapsto \strcon{a}, w \mapsto \strcon{aba}, \},\quad \tau_0 =\  (\Delta_0),\\
\Delta_0 &=\ &(\slen(w) \leq 3 \wedge x >_{\textsf{str}} \strcon{aabaa}) \vee (\slen(w) < 3 \wedge x \geq_{\textsf{str}} \strcon{aabaa}).
\end{array}
\]
The initial objective term value is $(3,\strcon{aabaa})$.

\begin{enumerate}
\item We can first apply \textsc{F-Split} to split the top-level disjunction in $\tau$.  And suppose we want to work on the second disjunct first.  This results in:
\[
\begin{array}{rcl}
\tau_1 & =\  & ( \slen(w) < 3 \wedge x \geq_{\textsf{str}} \strcon{aabaa},\slen(w) \leq 3 \wedge x >_{\textsf{str}} \strcon{aabaa})
\end{array}
\]
\noindent
while the other elements of the state are unchanged.

\item Now, suppose we want to do binary search on the length objective.  This can be done by again applying the $\textsc{F-Split}$ rule 
with the disjunction 
\((\slen(w) < 2 \wedge x \geq_{\textsf{str}} \strcon{aabaa}) 
  \lor 
  (2 \leq \slen(w) < 3 \wedge x \geq_{\textsf{str}} \strcon{aabaa})
\)
to get:
\[
\begin{array}{rcl}
    \tau_2 & =\ & ( \slen(w) < 2 \wedge x \geq_{\textsf{str}} \strcon{aabaa},2\le\slen(w) < 3 \wedge x \geq_{\textsf{str}} \strcon{aabaa},\\
&&\ \slen(w) \leq 3 \wedge x >_{\textsf{str}} \strcon{aabaa}).
\end{array}
\]
\noindent

\item
Both $\textsc{F-Split}$ and $\textsc{F-Sat}$ are applicable, but we follow the strategy of applying \textsc{F-Sat} after a split.  Suppose we get the new solution $\Model' = \{x \mapsto \strcon{b}, s \mapsto \epsilon, w \mapsto \strcon{b}\}$.  Then we have:
\[
\begin{array}{rcl}
   \Model_3 & = & \{x \mapsto \strcon{b}, s \mapsto \epsilon, w \mapsto \strcon{b}\},\quad \tau_3  =\ (\Delta_3),\\
    \Delta_3 & =\ & (\slen(w) \leq 1 \wedge x >_{\textsf{str}} \strcon{b})  \vee (\slen(w) < 1 \wedge x \geq_{\textsf{str}} \strcon{b}).
\end{array}
\]

\item Both $\textsc{F-Split}$ and $\textsc{F-Sat}$ are again applicable. Suppose that we switch now to linear search and thus again apply $\textsc{F-Sat}$, and suppose the new solution is $\Model' = \{x \mapsto \strcon{z}, s \mapsto \epsilon, w \mapsto \strcon{z}\}$.  This brings us to the state:
\[
\begin{array}{rcl}
   \Model_4 & =\ & \{x \mapsto \strcon{z}, s \mapsto \epsilon, w \mapsto \strcon{z}\},\quad \tau_4  =\ (\Delta_4),\\
    \Delta_4 & =\ & (\slen(w) \leq 1 \wedge x >_{\textsf{str}} \strcon{z})  \vee (\slen(w) < 1 \wedge x \geq_{\textsf{str}} \strcon{z}).
\end{array}
\]

    \item Now, $\textsc{Solve}(\phi \wedge ((\slen(w) \leq 1 \wedge x >_{\textsf{str}} \strcon{z}) \vee (\slen(w) < 1 \wedge x \geq_{\textsf{str}} \strcon{z}))) = \bot$. Indeed, $\slen(w) \neq 0$, since $0 \leq \slen (s) < \slen(w)$;
    if $\slen(w) = 1$, then $\slen(s) = 0$ and $\slen(x) = 1$, thus, $x \not >_{\textsf{str}} \strcon{z}$. Now $\textsc{F-Close}$ can derive the state:
    \[\langle \Model_5, \Delta_5, \tau_5\rangle = \langle \Model_4, \Delta_4 \wedge \neg\Delta_4, \emptyset \rangle\]

\item We have reached a saturated state, and $\Model_5$ is
  a Pareto optimal solution. 
  \qed 
\end{enumerate}
\end{example}

Optimization of objectives involving strings and integers (or strings and bitvectors) could be especially useful in security applications such as those mentioned in~\cite{Subramanian2017}.  Optimization could be used in such applications to ensure that a counter-example is as simple as possible, for example.

Examples of multi-objective problems unsupported by existing solvers include multiple Pareto problems with a single min/max query, Pareto-lexicographic multi-objective optimization, and single Pareto queries involving MinMax and MaxMin optimization (see, for example, \cite{inventions7030046,Lai2022,akinlanathesis}). Our framework offers immediate solutions to these problems.

As has repeatedly been the case in SMT research, when new capabilities are introduced, new applications emerge. We expect that will happen also for the new capabilities introduced in this paper. One possible application is the optimization of emerging technology circuit designs~\cite{delayconv}.

\subsection{Correctness}
\label{section:proofs:basic}

In this section, we establish correctness properties for \GO-derivations. 
Initially, we demonstrate that upon reaching a saturated state, 
the interpretation $\Model$ in that state is optimal.%
\footnote{%
\ifarxiv
Full proofs for the theorems in this section can be found in Appendix~\ref{appendix:proofs}. 
\else
Full proofs for the theorems in this section can be found in a longer version of this paper~\cite{GOMT}.
\fi
}
%
\begin{theorem} (Solution Soundness) 
\label{thm_basic_solution_soundness}
Let $\langle \Model, \Delta, \tau \rangle$ be a saturated state 
in a derivation for a \gomt problem \GO.
Then, $\Model$ is an optimal solution to \GO.
\end{theorem}

\begin{proof}(Sketch)
We show that in a saturated state $\tau = \emptyset$, and when $\tau = \emptyset$, $\phi \models \neg\Delta$. Then, we establish that $\Model$ is \GO-consistent, and that for any \GO-consistent $\T$-interpretation $\mathcal{J}$,
$\mathcal{J} \models \Delta$\; iff \;$\mathcal{J} <_{\GO} \Model$. This implies there is no $\mathcal{J}$ s.t. $\mathcal{J} \models \phi$ and $\mathcal{J} <_{\GO} \Model$, confirming $\Model$ as an optimal solution to \GO.
\qed
\end{proof}

In general, the calculus does not always have complete derivation strategies, for a variety of reasons.  It could be that the problem is unbounded, i.e., no optimal solutions exist along some branch.  Another possibility is that the order is not well-founded, and thus, an infinite sequence of improving solutions can be generated without ever reaching an optimal solution.
For the former, various checks for unboundedness can be used.  These are beyond the scope of this work, but some approaches are discussed in Trentin~\cite{patrickthesis}.  The latter can be overcome using a hybrid strategy when an optimization procedure exists (see Theorem~\ref{theorem_solve_termination}).
It is also worth observing that any derivation strategy 
is in effect an \emph{anytime procedure}:
forcibly stopping a derivation at any point  
yields (in the final state) the best solution found so far.  When an optimal solution exists and is unique, stopping early provides the best approximation up to that point of the optimal solution.

There are also fairly general conditions under which 
solution complete derivation strategies do exist.
We present them next.

\begin{definition}
A derivation strategy is \emph{progressive} if it (i) never halts in a non-saturated state and (ii) only uses \textsc{F-Split} a finite number of times in any derivation.
\end{definition}

\noindent
Let us again fix a \gomt problem $\GO = \langle t, \prec, \phi \rangle$.
Consider the set 
$A_t = \{ t^\Model \:\mid\: \Model \text{ is \GO-consistent} \}$,
collecting all values of $t$ in interpretations satisfying $\phi$.

\begin{theorem} (Termination)\label{theorem_basic_termination}
If $\prec$ is well-founded over $A_t$,
any progressive strategy reaches a saturated state.
\end{theorem}

\begin{proof}(Sketch)
We show that any derivation using a progressive strategy terminates when $\prec$ is well-founded. Subsequently, based on the definition of progressive, the final state must be saturated. 
\qed
\end{proof}

\begin{theorem} (Solution Completeness)\label{thm_basic_solution_completeness}
If $\prec$ is well-founded over $A_t$ and \GO has one or more optimal solutions,
every derivation generated by a progressive derivation strategy ends
with a saturated state containing one of them.
\end{theorem}

\begin{proof}
    The proof is a direct consequence of Theorem~\ref{thm_basic_solution_soundness} and Theorem~\ref{theorem_basic_termination}.
\qed
\end{proof}

\noindent
Solution completeness can also be achieved using an appropriate hybrid strategy.

\begin{theorem}
\label{theorem_solve_termination} 
If \GO has one or more optimal solutions and is not unbounded along any branch, then every derivation generated by a progressive hybrid strategy, where $\textsc{Solve}$ is replaced by $\textsc{Optimize}$ in $\textsc{F-Sat}$, ends with a saturated state containing one of them.
\end{theorem}

\begin{proof}(Sketch)
If $D$ is such a derivation, we note that $\textsc{F-Split}$ can only be applied a finite number of times in $D$ and consider the suffix of $D$ after the last application of $\textsc{F-Split}$.  
In that suffix, $\textsc{F-Close}$ can only be applied a finite number of times in a row, after which $\textsc{F-Sat}$ must be applied.  We then show that due to the properties of $\textsc{Optimize}$, this must be followed by either an application of $\textsc{F-Close}$ or a single application of $\textsc{F-Sat}$ followed by $\textsc{F-Close}$. Both cases result in saturated states.  The theorem then follows from Theorem~\ref{thm_basic_solution_soundness}. 
\qed
\end{proof}

\section{Related Work}
\label{section:related_work}
Various approaches for solving the OMT problem have been proposed. We summarize the key ideas below and refer the reader to Trentin~\cite{patrickthesis} for a more thorough survey.

The \emph{offline schema} employs an SMT solver as a black box for optimization search through incremental calls~\cite{PB-OMT12,Sebastiani:2015:OMT:2737801.2699915}, following linear- or binary-search strategies. Initial bounds on the objective function are given and  iteratively tightened after each call to the SMT solver.
In contrast, the \emph{inline schema} conducts the optimization search within the SMT solver itself~\cite{PB-OMT12,Sebastiani:2015:OMT:2737801.2699915}, integrating the optimization criteria into its internal algorithm. While the inline schema can be more efficient than the offline counterpart, it necessitates invasive changes to the solver and may not be possible for every theory.

\emph{Symbolic Optimization} optimizes multiple independent linear arithmetic objectives simultaneously~\cite{Li2014SymbolicOW}, seeking optimal solutions for each corresponding objective. This approach improves performance by sharing SMT search effort. It exists in both offline and inline versions, with the latter demonstrating superior performance.
%
Other arithmetic schemas 
combine simplex, branch-and-bound, and cutting-plane techniques within SMT solvers~\cite{LIA-OPT-Masters,10.1007/11814948_18}. A polynomial constraint extension has also been introduced~\cite{10.1007/978-3-319-09284-3_25}.

Theory-specific techniques address objectives involving pseudo-Booleans~\cite{PB-OMT12,Sebastiani:2015:OMT:2737801.2699915,DBLP:journals/corr/SebastianiT17,10.1007/978-3-642-12002-2_8}, bitvectors~\cite{Nadel2016BitVectorO,DBLP:journals/corr/abs-1905-02838}, bitvectors combined with floating-point arithmetic~\cite{DBLP:journals/jar/TrentinS21}, and nonlinear arithmetic~\cite{DBLP:conf/frocos/BigarellaCGIJRS21}.
Other related work includes techniques for lexicographic optimization~\cite{Bjrner2014ZM}, Pareto optimization~\cite{pareto,Bjrner2014ZM}, MaxSMT~\cite{FazekasBB18},
and All-OMT~\cite{patrickthesis}.

Our calculus is designed to capture all of these variations.  It directly corresponds to the offline schema, can handle both single- and multi-objective problems, and can 
integrate solvers with inline capabilities (including theory-specific ones) using the hybrid solving strategy. Efficient MaxSMT approaches~\cite{FazekasBB18} can also be mimicked in our calculus. These approaches systematically explore the search space by iteratively processing segments derived from unsat cores. Our calculus can instantiate these branches using the $\textsc{F-Split}$ rule, by first capturing unsat cores from calls to $\textsc{F-Close}$, and then using these cores to direct the search in the $\textsc{F-Split}$ rule.  

\section{Conclusion and Future Work}
\label{section:conclusion}

This paper introduces the Generalized OMT problem, a proper extension of the OMT problem.
It also provides a general setting for formalizing various approaches
for solving the problem in terms of a novel calculus for \gomt and proves its key correctness properties.
As with previous work on abstract transition systems for SMT~\cite{NieOT-JACM-06,KrsGoe-FROCOS-07,MouraB08,JovanovicM12},
this work establishes a framework for both theoretical exploration and 
practical implementations. 
The framework is general in several aspects: 
$(i)$ it is parameterized by the optimization order, which does not need to be
total;
$(ii)$ it unifies single- and multi-objective optimization problems in a single definition;
$(iii)$ it is theory-agnostic, making it applicable to any theory or combination of theories; and
$(iv)$ it provides a formal basis for understanding and exploring search strategies for generalized OMT.

In future work, we plan to explore an extension of the calculus to the generalized All-OMT problem.
We also plan to develop a concrete implementation of the calculus
in a state-of-the-art SMT solver and evaluate it experimentally against 
current OMT solvers.

\subsection*{Acknowledgements}
This work was funded in part by the Stanford Agile Hardware Center and by the National Science Foundation (grant 2006407).
\ifarxiv
\raggedbottom
\fi
\pagebreak


%
%
%
%
\bibliographystyle{splncs04}
\bibliography{main}


\appendix

\section{Example}
\label{appendix:examples}

\begin{example}[A single-objective \gomt]
Let $\GO:=\langle n, \bvusucc{4}, \phi \rangle$, where:
$\phi := y \bvulte{4} x \; \wedge \; 0010 \bvulte{4} x \; \wedge \; x \bvulte{4} 1000 \; \wedge \; 0001 \shl{4} n \bvulte{4} y \; \wedge \; y \bvult{4} 0001 \shl{4} (n \bvadd{2} 0001)\; \wedge \; n = x \bvsub{2} y $, and $n$, $x$, and $y$ are variables of sort \sbv{4}.\\

\noindent
We demonstrate execution of the \gomt calculus on $\GO$ with the binary search strategy over $n$. The binary search requires bounds on $n$. For simplicity, we derive the upper bound on $n$ from $\phi$. Specifically, $y \bvulte{4} x$, $x \bvulte{4} 1000$, $n = x \bvsub{4} y$, and $0001 \shl{4} n \bvulte{4} y$ imply that $n \bvulte{4} 0011$. Let $\Model_0:= \{n \mapsto 0001, x \mapsto 0011, y \mapsto 0010\},\; \Delta_0 := 0011 \bvugte{4} n \bvugt{4} 0001,\; \tau_0 := (n \bvugt{4} 0001),\; S_0 := \emptyset$.
The initial objective term value is $0001$. We provide a possible execution of the calculus below. 
\begin{enumerate}[leftmargin=1em]
\item We apply the $\textsc{F-Split}$ rule.
\[
\inferrule*[Left=F-Split]{
  \tau \neq \emptyset \;\;\; \psi = n \bvugt{4} 0001 \\
  \phi \!\models \!(0010\! \bvugte{4} \!n \!\bvugt{4} \!0001 \!\vee\! 0011 \!\bvugte{4}\! n \!\bvugt{4}\! 0010 \!\liff\! \psi)
}{
 \tau := (0011 \!\bvugte{4} \!n\! \bvugt{4} \!0010, 0010\! \bvugte{4}\! n\! \bvugt{4}\! 0001)
}
\]

\item From the first branch $n = 0011$ and must hold $y = 1000$, however this conflicts with the constraints $n = x \bvsub{4} y$ and $x \bvulte{4} 1000$. Thus, the next applicable rule is $\textsc{F-Close}$.

\[
\inferrule*[Left=F-Close]{
  \tau \neq \emptyset \quad \psi =  0011\! \bvugte{4}\! n\! \bvugt{4} \!0010 \quad \textsc{Solve}(\phi \!\wedge\! \psi) \models \bot 
}{
  \Delta := 0011\bvugte{4} \!n\! \bvugt{4} \!0001 \!\wedge \!\neg (0011 \!\bvugte{4}\! n\! \bvugt{4}\! 0010), \\ 
  \tau := (0010\! \bvugte{4}\! n\! \bvugt{4}\! 0001)
}
\]
\noindent
The search space in the postcondition is simplified to $\Delta = 0010 \bvugte{4} n \bvugt{4} 0001$.

\item Now, $\textsc{Top}(\tau)$ implies $n = 0010$ and $\phi \wedge \textsc{Top}(\tau) \not\models \bot$. Consequently, either rule $\textsc{F-Split}$ or $\textsc{F-Sat}$ is applicable. Let us apply $\textsc{F-Sat}$ and derive a better solution of $\phi$, $\Model ' = \{n \mapsto 0010, x \mapsto 0111, y \mapsto 0101\} $.

\[
\inferrule*[Left=F-Sat]{
  \tau \neq \emptyset \quad \psi = 0010 \!\bvugte{4} \!n \!\bvugt{4} \!0001 \quad 
  \textsc{Solve}(\phi \wedge \psi) =\Model' \quad \Model' \neq \bot\\
  \Delta' =  (0010 \!\bvugte{4}\! n\! \bvugt{4}\! 0001)\! \wedge \!(n\! \bvugt{4} \!0010)
}{
  \Model := \{n \mapsto 0010, x \mapsto 0111, y \mapsto 0101\}, \; \Delta := \bot, \; \tau:= (\bot)
}
\]
\item Next, $\textsc{F-Close}$ is applicable since $\tau = (\bot)$ and $\phi \wedge \bot \models \bot$.
\[
\inferrule*[Left=F-Close]{
  \tau \neq \emptyset \quad \psi = \bot \quad \textsc{Solve}(\phi \wedge \psi) \models \bot 
}{
  \Delta := \bot \wedge \neg (\bot), \; \tau := \emptyset
}
\]
And $\Model = \{n \mapsto 0010, x \mapsto 0101, y \mapsto 0100\}$ is a solution to $\GO$.
\end{enumerate}
\end{example}

\section{Correctness}
\label{appendix:proofs}
\renewcommand{\thetheorem}{\arabic{theorem}}
\setcounter{theorem}{0}
\setcounter{lemma}{0}
\setcounter{corollary}{0}

In this section, we provide the details for the theorems in Section~\ref{section:proofs:basic}.  We again fix a \gomt problem $\GO=\langle t, \prec, \phi \rangle$. 

\begin{lemma}\label{lemma_model_omt_consistent}($\GO$-consistency)
Each element of the solution sequence of a $\GO$-derivation is $\GO\!$-consistent.
\end{lemma}
\begin{proof}
    Let $I$ be the solution sequence, and 
    let $\Model_0=\tops(I)$.  By definition of derivations and initial states, we have that $\Model_0 \models \phi$. Now, consider some $\Model\in I$ such that $\Model\not=\Model_0$.  Observe that it must have been introduced by an application of 
    $\textsc{F-Sat}$, since only this rule changes the solution in the state. But \textsc{F-Sat} explicitly invokes the $\textsc{Solve}$ function on a formula that conjunctively includes $\phi$. Thus, $\Model\models\phi$.
    \qed
\end{proof}

\begin{lemma}\label{lemma_transitive}(Transitivity)
The operator $<_{\GO}$ is transitive.
\end{lemma}

\begin{proof}
    The proof follows from the transitivity property of the strict partial order $\prec$. Suppose
    $\Model<_{\GO}\Model'$ and $\Model'<_{\GO}\Model''$. 
    Then,  $\Model,\; \Model',\;\Model''$ are $\GO$-consistent, $t^{\Model} \prec t^{\Model'}$, and $t^{\Model'} \prec t^{\Model}$. 
    Since $\prec$ is transitive,
    $t^{\Model} \prec t^{\Model}$, and since $\Model$ and $\Model''$ are $\GO$-consistent, we conclude: $\Model<_{\GO}\Model''$. \qed
\end{proof}

\noindent
Lemma~\ref{lemma_tau_eq_Delta} and Lemma~\ref{lemma_tau_disjunction} capture the fact that updates to $\tau$ preserve the set of $\GO$-consistent interpretations in the explored search space $\Delta$. Lemma~\ref{lemma_tau_eq_Delta} ensures that each sub-formula in $\tau$ represents a part of the overall search space. Lemma~\ref{lemma_tau_disjunction} ensures that $\tau$ always covers the entire search space.

\begin{lemma}\label{lemma_tau_eq_Delta}(Non-expansiveness)
Let $\langle \Model, \Delta, \tau\rangle$ be a state in a $\GO$-derivation.  Then, for each $\tau_i\in\tau$, $\phi \models (\tau_i \limpl \Delta)$. 
\end{lemma}

\begin{proof} 
The proof is by induction on derivations. \\

\noindent
(Base case) The lemma holds trivially in the initial state where $\Delta_0:=\better(\Model_0)$ and $\tau_0:=(\better(\Model_0))$.\\

\noindent
(Inductive case)
Let $\langle \Model', \Delta', \tau' \rangle$ be a state and assume that for each $\tau'_i\in\tau$, $\phi \models (\tau'_i \limpl \Delta')$.  Let   $\langle \Model'', \Delta'', \tau''\rangle$ be the next state.  We show that for each $\tau''_i\in\tau''$, $\phi \models (\tau''_i \limpl \Delta'')$.

\begin{itemize}
    \item $\textsc{F-Split}$ modifies $\tau'$ by replacing $\textsc{Top}(\tau')$ with $(\psi'_1,\dots,\psi'_k)$, where $\phi\models\bigvee_{j=1}^k \psi'_j \liff \textsc{Top}(\tau')$. $\Delta'$ remains unmodified. By the induction hypothesis, $\phi \models (\textsc{Top}(\tau') \limpl \Delta')$. It follows that $\phi \models (\bigvee_{j=1}^k \psi'_j \limpl \Delta')$, and thus $\phi \models (\psi'_j \limpl \Delta')$ for each $1 \leq j \leq k$.  The rest of $\tau'$ is unchanged, so the property is preserved.

    \item $\textsc{F-Sat}$ sets $\tau'' := (\Delta'')$. Clearly (as in the base case), $\phi \models \Delta'' \limpl \Delta''$.

    \item $\textsc{F-Close}$ pops $\textsc{Top}(\tau')$ from $\tau'$ and sets $\Delta'':=\Delta' \wedge \neg \textsc{Top}(\tau')$. The premise tells us that $\phi \wedge \textsc{Top}(\tau')$ is unsatisfiable, meaning that $\phi \models \neg\tops(\tau')$.  It follows that $\phi \models \Delta' \liff \Delta''$. Now, suppose $\tau''_i \in \tau''$.  Then, also $\tau''_i\in\tau'$.  By the induction hypothesis, $\phi \models (\tau''_i \limpl \Delta')$, but then also $\phi \models (\tau''_i \limpl \Delta'')$.
\qed
\end{itemize}

\end{proof}

\begin{lemma}\label{lemma_tau_disjunction}(Non-omissiveness)
Let $\langle \Model, \Delta, \tau \rangle$ be a state in a $\GO$-derivation.  If $\tau = (\tau_1,\dots,\tau_m)$, then $\phi \models \bigvee_{i=1}^m \tau_i \liff \Delta$.
\end{lemma}

\begin{proof} Note that we define $\bigvee_{i=m}^n \tau_i$ := $\false$ if $m>n$.
The proof of the lemma is by induction. \\

\noindent
(Base case) In the initial state, $\Delta_0 = \better(\Model_0)$, $\tau_0 = \{\better(\Model_0)\}$, and clearly, $\phi \models \better(\Model_0) \liff \better(\Model_0)$.\\

\noindent
(Inductive case)  Let $\langle \Model', \Delta', \tau' \rangle$ be a state with $\tau'\!=\!(\tau'_1,\dots,\tau'_{m'})$. Assume $\phi \models \bigvee_{i=1}^{m'}\tau'_i \liff \Delta'$.  Let  $ \langle \Model'', \Delta'', \tau'' \rangle$ be the next state with $\tau''\!=\!(\tau''_1,\!\dots\!,\tau''_{m''})$.  We show $\phi \!\models \!\bigvee_{i=1}^{m''} \tau''_i \!\liff\! \Delta''$.
\begin{itemize}
    \item $\textsc{F-Split}$ modifies $\tau'$ by replacing $\textsc{Top}(\tau')$ with $(\psi'_1,\dots,\psi'_k)$, where $\phi\models\bigvee_{i=1}^k \psi'_j \liff \textsc{Top}(\tau')$. $\Delta'$ remains unmodified.  Thus, $\phi\models \bigvee_{i=1}^{m''} \tau''_i \liff \bigvee_{i=1}^{m'} \tau'_i \liff \Delta' \liff \Delta''$.

    \item $\textsc{F-Sat}$ sets $\tau'' = \{\Delta''\}$. And clearly, $\phi \models \Delta'' \liff \Delta''$.

    \item $\textsc{F-Close}$ pops $\textsc{Top}(\tau')$ from $\tau'$ and sets $\Delta'':=\Delta' \wedge \neg \textsc{Top}(\tau')$. The premise tells us that $\phi \wedge \textsc{Top}(\tau')$ is unsatisfiable, meaning that $\phi \models \neg\tops(\tau')$.
    From this, it follows that: $(i)$ $\phi \models \bigvee_{i=1}^{m'} \tau'_i \liff \bigvee_{i=2}^{m'} \tau'_i$; and $(ii)$ $\phi \models \Delta' \liff \Delta' \wedge \neg \tops(\tau')$. From the induction hypothesis, together with (i) and (ii), we get $\phi \models \bigvee_{i=2}^{m'} \tau'_i \liff \Delta'\wedge \neg \tops(\tau')$. But this is exactly $\phi \models \bigvee_{i=1}^{m''} \tau''_i \liff \Delta''$.
\qed
    
\end{itemize}

\end{proof}    

\noindent
Corollary~\ref{corollary_Delta_eq_bot} states that if there are no branches to explore, the search space is empty.

\begin{corollary}\label{corollary_Delta_eq_bot}
    If $\langle \Model, \Delta, \tau \rangle$ is a state, and $\tau = \emptyset$, then $\phi \models \neg \Delta$.
\end{corollary}

\begin{proof}
    Proof by Lemma~\ref{lemma_tau_disjunction}: if $\tau=\emptyset$, then $m = 0$ and $\bigvee_{i=1}^m \tau_i = \false$, so $\phi \models \false \liff \Delta$ and thus $\phi \models \neg\Delta$.
    \qed
\end{proof}

\noindent
Intuitively, Lemma~\ref{lemma_invariant} below states that the search space always contains better (according to $<_\GO$) solutions, if any, than the current solution, and any better solution than the current one is guaranteed to be in the search space.

\begin{lemma}\label{lemma_invariant}(Always Better)
Let $\langle \Model, \Delta, \tau \rangle$ be a state in a $\GO$-derivation, and let $\mathcal{J}$ be a $\T$-interpretation. 
If $\mathcal{J}$ is $\GO$-consistent, then 
$\mathcal{J} \models \Delta$\; iff \;$\mathcal{J} <_{\GO} \Model$.
\end{lemma}

\begin{proof}
The proof of the lemma is by induction. \\

\noindent
(Base case) In the initial state, $\Delta_0 = \better(\Model_0)$, and $\Model_0$ is $\GO$-consistent by Lemma~\ref{lemma_model_omt_consistent}.  By the definition of $\better$, if $\mathcal{J}$ is a $\GO$-consistent interpretation, then $\mathcal{J}\models \better(\Model_0)$ iff $\mathcal{J} <_\GO \Model_0$.\\

\noindent
(Inductive case)  Let $\langle \Model', \Delta', \tau' \rangle$, be a state such that if $\mathcal{J}$ is a $\GO$-consistent interpretation, then $\mathcal{J}\models \Delta'$ iff $\mathcal{J} <_\GO \Model'$.
Let  $\langle \Model'', \Delta'', \tau''\rangle$ be the next state.  We show that if $\mathcal{J}$ is a $\GO$-consistent interpretation, then $\mathcal{J}\models \Delta''$ iff $\mathcal{J} <_\GO \Model''$.

\begin{itemize}

    \item $\textsc{F-Split}$ does not modify $\Model'$ or $\Delta'$.

    \item $\textsc{F-Sat}$ Let $\mathcal{J}$ be a $\GO$-consistent interpretation.\\
    
    (Forward direction) Assume $\mathcal{J}\models \Delta''$.  We know that $\Delta'' = \Delta' \wedge \better(\Model'')$.  Thus $\mathcal{J} \models \better(\Model'')$.  By the definition of $\better$, $\mathcal{J} <_\GO \Model''$.\\
    
    (Backward direction) Assume that $\mathcal{J} <_\GO \Model''$. Then, (i) $\Model''$ is $\GO$-consistent by the definition of $<_\GO$, and (ii) $\mathcal{J} \models \better(\Model'')$ by the definition of $\better$.
    From the  $\textsc{F-Sat}$ rule, we know that $\Delta'' = \Delta' \wedge \better(\Model'')$. Thus, by (ii), it remains to show that $\mathcal{J} \models \Delta'$.
    By (i) and Lemma~\ref{lemma_tau_eq_Delta}, we have $\Model'' \models \textsc{Top}(\tau') \limpl \Delta'$. By the premise of the $\textsc{F-Sat}$ rule and the definition of \solve, we have $\Model'' \models \textsc{Top}(\tau')$, so it follows that $\Model'' \models \Delta'$.  Now, by the induction hypothesis, we have $\Model'' <_{\GO} \Model'$, so by Lemma~\ref{lemma_transitive}, we have $\mathcal{J} <_{\GO} \Model'$.  Then, again by the induction hypothesis, we have $\mathcal{J} \models \Delta'$.

    \item $\textsc{F-Close}$ does not modify $\Model'$. Thus, $\Model'' = \Model'$.  Let $\mathcal{J}$ be a $\GO$-consistent interpretation.\\
    
    (Forward direction) Assume $\mathcal{J}\models \Delta''$. From the  $\textsc{F-Close}$ rule, we know that $\Delta'' = \Delta' \wedge \neg \tops(\tau')$. Thus, $\mathcal{J} \models \Delta'$. Then, by the induction hypothesis, $\mathcal{J} <_{\GO} \Model'$.  But $\Model' = \Model''$, so $\mathcal{J} <_\GO \Model''$\\

    (Backward direction) Assume that $\mathcal{J} <_\GO \Model''$ and thus $\mathcal{J} <_\GO \Model'$, since $\Model''=\Model'$. Then, $\mathcal{J} \models \Delta'$ by the induction hypothesis. Now, from the premise of $\textsc{F-Close}$, we know that $\Delta' \wedge \tops(\tau')$ is unsatisfiable. Since $\mathcal{J} \models \Delta'$, we must have that $\mathcal{J} \models \neg \tops(\tau')$. Finally, since $\Delta'' = \Delta' \wedge \neg \tops(\tau')$, we have $\mathcal{J} \models \Delta''$.\qed
\end{itemize}    
\end{proof}

\begin{theorem} (Solution Soundness) Let $\langle \Model, \Delta, \tau \rangle$ be a saturated state in a derivation for a \gomt problem $\GO\!$. Then, $\Model$ is an optimal solution to $\GO$.
\label{thm_basic_solution_soundness_appdx}
\end{theorem}

\begin{proof}
       If $\tau\not=\emptyset$, $\textsc{F-Split}$ is applicable. 
       Thus, we must have $\tau=\emptyset$ in a saturated state. 
       By Lemma~\ref{lemma_model_omt_consistent}, $\Model$ is $\GO$-consistent, and, thus, $\Model \models \phi$.
       Since $\tau = \emptyset$, by Corollary~\ref{corollary_Delta_eq_bot}, $\phi \models \neg\Delta$. Consequently, by Lemma~\ref{lemma_invariant}, there is no $\mathcal{J}$, such that $\mathcal{J} \models \phi$ and $\mathcal{J} <_{\GO} \Model$. Thus, $\Model$ is an optimal solution to $\GO$. \qed
\end{proof}

\begin{lemma}\label{lemma_F-Sat_better} (Improvement Step)
    Let $\mathcal{S} = \langle \Model, \Delta, \tau \rangle$ and $\mathcal{S'} = \langle \Model', \Delta', \tau' \rangle$ be two states. 
    Suppose that $\mathcal{S'}$ is the result of applying \textsc{F-Sat} to $\mathcal{S}$.  Then $\Model' <_\GO \Model$.
\end{lemma}

\begin{proof}
     From the premise of the $\textsc{F-Sat}$ rule, we know that $\Model' \models \phi \wedge \textsc{Top}(\tau)$. Then, by Lemma~\ref{lemma_tau_eq_Delta}, we have $\Model' \models \Delta$. It follows from Lemmas~\ref{lemma_model_omt_consistent} and~\ref{lemma_invariant} that $\Model' <_{\GO} \Model$.
     \qed
\end{proof}

Now, recall that $A_t = \{ t^\Model \,\mid\, \Model \text{ is $\GO$-consistent} \}$ is the set of all values that $t$ has in $\T$-interpretations satisfying $\phi$. The termination and solution completeness arguments follow.

\begin{theorem} (Termination)~\label{thm_basic_termination_appdx}
    If $\prec$ is well-founded over $A_t$,
any progressive strategy reaches a saturated state.
\end{theorem}

\begin{proof}
Since a progressive strategy uses the \textsc{F-Split} rule only a finite number of times, any infinite derivation must use either the \textsc{F-Sat} rule or the \textsc{F-Close} rule an infinite number of times.

   Let $D$ be such a derivation and suppose that \textsc{F-Sat} is used an infinite number of times along $D$.  If $s$ is the sequence consisting only of states of $D$ resulting from applications of \textsc{F-Sat}, then let $I$ be the corresponding solution sequence. Since \textsc{F-Close} does not change $\Model$, we have that if $\Model$ and $\Model'$ are consecutive solutions in $I$, i.e., solutions in states resulting from applications of \textsc{F-Sat}, then $\Model' <_\GO \Model$ by Lemma~\ref{lemma_F-Sat_better}. But then, by the definition of $<_\GO$, there must be an infinite descending $\prec$-chain in $A_t$, contradicting the assumption that $\prec$ is well-founded on $A_t$. 
   Thus, $D$ contains only a finite number of applications of \textsc{F-Sat}.
   
Since both \textsc{F-Split} and \textsc{F-Sat} are applied only a finite number 
of times in $D$, eventually no formulas get added to $\tau$.
Since \textsc{F-Close} applies only to states with a non-empty $\tau$
and reduces its size by one,
this rule too can be applied at most finitely many times in $D$.
Hence, $D$ cannot be infinite.
   
Thus, when $\prec$ is well-founded over $A_t$, any derivation using a progressive strategy terminates.
By the definition of progressive, the final state must be saturated. \qed
\end{proof}

\begin{theorem} (Solution Completeness)
    If $\prec$ is well-founded over $A_t$ and $\GO$ has one or more optimal solutions, every derivation generated by a progressive derivation strategy ends with a saturated state containing one of them.\label{thm_basic_solution_completeness_appdx}
\end{theorem}

\begin{proof}
    The proof is a direct consequence of Theorem~\ref{thm_basic_solution_soundness_appdx} and Theorem~\ref{thm_basic_termination_appdx}.\qed
\end{proof}

\begin{theorem}
\label{theorem_solve_termination_apdx} 
If $\GO$ has one or more optimal solutions and is not unbounded along any branch, then every derivation generated by a progressive hybrid strategy, where $\textsc{Solve}$ is replaced by $\textsc{Optimize}$ in $\textsc{F-Sat}$, ends with a saturated state containing one of them.
\end{theorem}

\begin{proof}

We show that every such derivation must terminate in a saturated state.  The result then follows from Theorem~\ref{thm_basic_solution_soundness}.

Consider such a derivation $D$.  A progressive strategy only uses the $\textsc{F-Split}$ rule a finite number of times.  Let $\mathcal{S}$ be the state resulting from the last application of $\textsc{F-Split}$ in $D$ (or the initial state if $\textsc{F-Split}$ is never used), and let $D'$ be the sequence that is the suffix of the original derivation starting with $\mathcal{S}$.  Because $\textsc{F-Close}$ reduces the size of $\tau$ by one, $\textsc{F-Close}$ can only be applied a finite number of times in a row starting from $\mathcal{S}$.  Let $\hat{\mathcal{S}}$ be the first state in $D'$ to which $\textsc{F-Close}$ is not applied.  Either $\hat{\mathcal{S}}$ is saturated, in which case we are done, or $\textsc{F-Sat}$ is applied to $\hat{\mathcal{S}}$, resulting in a state $\mathcal{S}' = \langle \Model', \Delta', (\Delta')\rangle$, where $\Delta' = \Delta \wedge \better(\Model')$ for some $\Delta$.  If, in this state, $\textsc{F-Close}$ is applied, then the resulting state has $\tau=\emptyset$ and is thus saturated.  On the other hand, suppose $\textsc{F-Sat}$ is applied to this state, resulting in $\mathcal{S}'' = \langle \Model'', \Delta'', \tau''\rangle$, where $\tau'' = (\Delta'')$.

By the definition of $\textsc{Optimize}$, we know that $\Model''$ is an optimal solution satisfying $\phi \wedge \Delta \wedge \better(\Model')$.  In particular, this means that $\phi \wedge \Delta \wedge \better(\Model') \wedge \better(\Model'')$ is not satisfiable.  But $\Delta'' = \Delta \wedge \better(\Model') \wedge \better(\Model'')$, and $\tau'' = (\Delta'')$, so this implies that $\textsc{F-Sat}$ cannot be applied to $\mathcal{S}''$ and $\textsc{F-Close}$ can.  Because the derivation is progressive, it thus must apply $\textsc{F-Close}$ next, resulting in a state with $\tau=\emptyset$, a saturated state. \qed
\end{proof}

\end{document}